\documentclass[12pt,a4paper]{article}
\usepackage{authblk}
\usepackage[utf8]{inputenc}
\usepackage[T1]{fontenc}
\usepackage{amsmath, amssymb, amsthm}
\usepackage{geometry}
\usepackage{braket}   
\usepackage{amsmath,amssymb,amsfonts}
\usepackage{graphicx}
\usepackage{hyperref}
\usepackage[utf8]{inputenc}
\usepackage[english]{babel}
\usepackage{cite}
\usepackage{mathrsfs}
\usepackage{amsthm}
\usepackage{enumitem}
\usepackage{nccmath}

\newtheorem{definition}{Definition}
\newtheorem{question}[definition]{Question}
\newtheorem{lemma}[definition]{Lemma}
\newtheorem{remark}[definition]{Remark}
\newtheorem{theorem}[definition]{Theorem}
\newtheorem{example}[definition]{Example}
\newtheorem{proposition}[definition]{Proposition}

\newtheorem{corollary}[definition]{Corollary}
\newtheorem{conjecture}[definition]{Conjecture}

\newtheorem{memo}[definition]{Memo}

\def\squareforqed{\hbox{\rlap{$\sqcap$}$\sqcup$}}
\def\qed{\ifmmode\squareforqed\else{\unskip\nobreak\hfil
\penalty50\hskip1em\null\nobreak\hfil\squareforqed
\parfillskip=0pt\finalhyphendemerits=0\endgraf}\fi}
\def\endenv{\ifmmode\;\else{\unskip\nobreak\hfil
\penalty50\hskip1em\null\nobreak\hfil\;
\parfillskip=0pt\finalhyphendemerits=0\endgraf}\fi}
\def\Dbar{\leavevmode\lower.6ex\hbox to 0pt
{\hskip-.23ex\accent"16\hss}D}
\makeatletter
\def\url@leostyle{%
  \@ifundefined{selectfont}{\def\UrlFont{\sf}}{\def\UrlFont{\small\ttfamily}}}
\makeatother

\def\bcj{\begin{conjecture}}
\def\ecj{\end{conjecture}}
\def\bcr{\begin{corollary}}
\def\ecr{\end{corollary}}
\def\bd{\begin{definition}}
\def\ed{\end{definition}}
\def\bea{\begin{eqnarray}}
\def\eea{\end{eqnarray}}
\def\beq{\begin{equation}}
\def\eeq{\end{equation}}
\def\bal{\begin{aligned}}
\def\eal{\end{aligned}}
\def\bem{\begin{enumerate}}
\def\eem{\end{enumerate}}
\def\bex{\begin{example}}
\def\eex{\end{example}}
\def\bim{\begin{itemize}}
\def\eim{\end{itemize}}
\def\bl{\begin{lemma}}
\def\el{\end{lemma}}
\def\bma{\begin{bmatrix}}
\def\ema{\end{bmatrix}}
\def\bpf{\begin{proof}}
\def\epf{\end{proof}}
\def\bpp{\begin{proposition}}
\def\epp{\end{proposition}}
\def\bqu{\begin{question}}
\def\equ{\end{question}}
\def\br{\begin{remark}}
\def\er{\end{remark}}
\def\bt{\begin{theorem}}
\def\et{\end{theorem}}
\def\bmm{\begin{memo}}
\def\emm{\end{memo}}

\def\btb{\begin{tabular}}
\def\etb{\end{tabular}}

\newcommand{\nc}{\newcommand}

\def\s{\sigma}

\def\ps{\psi}

\nc{\bbA}{\mathbb{A}} \nc{\bbB}{\mathbb{B}} \nc{\bbC}{\mathbb{C}}
 \nc{\bbD}{\mathbb{D}} \nc{\bbE}{\mathbb{E}} \nc{\bbF}{\mathbb{F}}
 \nc{\bbG}{\mathbb{G}} \nc{\bbH}{\mathbb{H}} \nc{\bbI}{\mathbb{I}}
 \nc{\bbJ}{\mathbb{J}} \nc{\bbK}{\mathbb{K}} \nc{\bbL}{\mathbb{L}}
 \nc{\bbM}{\mathbb{M}} \nc{\bbN}{\mathbb{N}} \nc{\bbO}{\mathbb{O}}
 \nc{\bbP}{\mathbb{P}} \nc{\bbQ}{\mathbb{Q}} \nc{\bbR}{\mathbb{R}}
 \nc{\bbS}{\mathbb{S}} \nc{\bbT}{\mathbb{T}} \nc{\bbU}{\mathbb{U}}
 \nc{\bbV}{\mathbb{V}} \nc{\bbW}{\mathbb{W}} \nc{\bbX}{\mathbb{X}}
 \nc{\bbZ}{\mathbb{Z}}

 \nc{\bA}{{\bf A}} \nc{\bB}{{\bf B}} \nc{\bC}{{\bf C}}
 \nc{\bD}{{\bf D}} \nc{\bE}{{\bf E}} \nc{\bF}{{\bf F}}
 \nc{\bG}{{\bf G}} \nc{\bH}{{\bf H}} \nc{\bI}{{\bf I}}
 \nc{\bJ}{{\bf J}} \nc{\bK}{{\bf K}} \nc{\bL}{{\bf L}}
 \nc{\bM}{{\bf M}} \nc{\bN}{{\bf N}} \nc{\bO}{{\bf O}}
 \nc{\bP}{{\bf P}} \nc{\bQ}{{\bf Q}} \nc{\bR}{{\bf R}}
 \nc{\bS}{{\bf S}} \nc{\bT}{{\bf T}} \nc{\bU}{{\bf U}}
 \nc{\bV}{{\bf V}} \nc{\bW}{{\bf W}} \nc{\bX}{{\bf X}}
 \nc{\bZ}{{\bf Z}}

\nc{\cA}{{\cal A}} \nc{\cB}{{\cal B}} \nc{\cC}{{\cal C}}
\nc{\cD}{{\cal D}} \nc{\cE}{{\cal E}} \nc{\cF}{{\cal F}}
\nc{\cG}{{\cal G}} \nc{\cH}{{\cal H}} \nc{\cI}{{\cal I}}
\nc{\cJ}{{\cal J}} \nc{\cK}{{\cal K}} \nc{\cL}{{\cal L}}
\nc{\cM}{{\cal M}} \nc{\cN}{{\cal N}} \nc{\cO}{{\cal O}}
\nc{\cP}{{\cal P}} \nc{\cQ}{{\cal Q}} \nc{\cR}{{\cal R}}
\nc{\cS}{{\cal S}} \nc{\cT}{{\cal T}} \nc{\cU}{{\cal U}}
\nc{\cV}{{\cal V}} \nc{\cW}{{\cal W}} \nc{\cX}{{\cal X}}
\nc{\cZ}{{\cal Z}}

\nc{\hA}{{\hat{A}}} \nc{\hB}{{\hat{B}}} \nc{\hC}{{\hat{C}}}
\nc{\hD}{{\hat{D}}} \nc{\hE}{{\hat{E}}} \nc{\hF}{{\hat{F}}}
\nc{\hG}{{\hat{G}}} \nc{\hH}{{\hat{H}}} \nc{\hI}{{\hat{I}}}
\nc{\hJ}{{\hat{J}}} \nc{\hK}{{\hat{K}}} \nc{\hL}{{\hat{L}}}
\nc{\hM}{{\hat{M}}} \nc{\hN}{{\hat{N}}} \nc{\hO}{{\hat{O}}}
\nc{\hP}{{\hat{P}}} \nc{\hR}{{\hat{R}}} \nc{\hS}{{\hat{S}}}
\nc{\hT}{{\hat{T}}} \nc{\hU}{{\hat{U}}} \nc{\hV}{{\hat{V}}}
\nc{\hW}{{\hat{W}}} \nc{\hX}{{\hat{X}}} \nc{\hZ}{{\hat{Z}}}

\nc{\hn}{{\hat{n}}}

\nc{\as}{{\cal AS}}
	\nc{\app}{{\cal AP}}

\def\dim{\mathop{\rm Dim}}

\def\lin{\mathop{\rm span}}

\newcommand{\proj}[1]{| #1\rangle\!\langle #1 |}

\newcommand{\abs}[1]{|#1|}

\newcommand{\keywords}[1]{%
  \par\noindent\textbf{Keywords:} #1%
}

\allowdisplaybreaks[4]

\hypersetup{colorlinks,linkcolor={blue},citecolor={blue},urlcolor={red}}

\begin{document}

\title{Entanglement Distillation of some Rank-Five Symmetric NPT States
       in Two-Qutrit Systems}

\author{
  Yuwei Lei,
  Zihua Song\textsuperscript{1}\textsuperscript{*},
  Lin Chen\textsuperscript{1}\textsuperscript{*},
  and
  Mingju Liu\textsuperscript{1}\\
  \vspace{0.1cm}
  {\small \textsuperscript{1} LMIB (Beihang University), Ministry of Education, and School of Mathematical Sciences, Beihang University, Beijing 100191, China;}
}
\footnotetext[1]{* Corresponding authors: linchen@buaa.edu.cn (Lin Chen)}

\date{\today}

\maketitle

\begin{abstract}
Entanglement distillation is a fundamental task in quantum information processing. In this work, we investigate the distillability properties of a class of two-qutrit symmetric NPT states of rank five. We resolve the 1-distillability problem for this class by proving that the previously open interval of the eigenvalue parameter is 1-undistillable. For the 2-distillability, we uncover a structural obstruction showing that no Schmidt-rank-two vector has a negative expectation in the relevant subspace. We also perform numerical investigations to explore the 2-distillability beyond this obstruction.
\end{abstract}

\keywords{Entanglement distillation, symmetric state, 1-distillability}





\section{Introduction}

Quantum information technology is currently recognized as a key direction of the Fourth Industrial Revolution. Entanglement is now widely regarded as an essential resource for quantum information processing tasks~\cite{horodecki2009quantum}. However, environmental noise and experimental imperfections typically produce mixed entangled states, whose quality is often insufficient for direct use.
Entanglement distillation addresses this issue by converting many copies of a mixed entangled state into a pure entangled state asymptotically through local operations and classical communication (LOCC)~\cite{bennett1996concentrating}. The resulting pure states enable long-distance quantum communication and distributed quantum computing. However, not all mixed entangled states are distillable~\cite{divincenzo2000evidence}. Determining which states are distillable~\cite{horodecki1997inseparable} is a fundamental problem and has been listed among the five important open questions in quantum information~\cite{horodecki2020five}.

The distillability of an entangled state is closely connected to its matrix rank. Rank-two, rank-three, and rank-four NPT entangled states are known to be distillable~\cite{horodecki2003rank, chen2008rank, chen2016distillability}. The rank-five case is significantly more complex due to its increased structural degrees of freedom. For this rank, the relationship between the distribution of negative eigenvalues of the partial transpose and distillability remains unclear, although existing studies have characterized such NPT states by counting negative eigenvalues~\cite{rana2013negative, johnston2013nonpositive} or via Hermitian matrix inertia~\cite{shen2020inertias}.
Various criteria, including the positive partial transpose (PPT) criterion~\cite{peres1996separability, horodecki1996separability}, the reduction criterion~\cite{horodecki1999reduction}, and entanglement witnesses~\cite{terhal2000bell}, have been developed to address distillability. In bipartite systems, all two-qubit entangled states and all \(2\times n\) NPT entangled states are distillable~\cite{horodecki1997inseparable}. However, higher-dimensional systems can contain PPT entangled states, also known as bound entangled states, which exhibit limited distillability~\cite{horodecki1997separability}. Werner states represent a key class of NPT states for investigating general distillability~\cite{werner1989quantum}. Nevertheless, the asymptotic distillation framework introduces multi-copy scenarios~\cite{watrous2004many}, while the inherent complexity of LOCC operations~\cite{chitambar2014everything} and the use of tools such as the Choi-Jamio{\l}kowski isomorphism~\cite{choi1975completely} add further mathematical challenges to the determination of distillability.
A recent study~\cite{song2026entanglement} systematically investigated the 1-distillability of a family of rank-five symmetric bipartite qutrit entangled states. By imposing equality conditions on four of the five eigenvalues, five classes of parametrized states were constructed. For four of these classes, the 1-distillability was fully established over the entire parameter range. However, in the remaining class, where the distinguished eigenvalue is associated with the eigenvector $|e_5\rangle$, distillability remains unresolved for a specific parameter interval, namely $\left[\frac{24\sqrt{2} - 33}{7},\ \frac{33 - 12\sqrt{6}}{25}\right)$. Within this interval, the state is known to be NPT, yet its 1-distillability has not been determined. 
A single example at $\lambda_5=1/7$ was shown to be 1-undistillable, but the general case remains open. 

In this paper, we build upon the above line of research by investigating the distillability of rank-five NPT entangled states $\rho$ studied in the recent work~\cite{song2026entanglement}. We begin with Definition~\ref{def:distill}, which formalizes the notion of $n$-distillability under LOCC and serves as the central concept of this work. Lemma~\ref{lem:principal_minors} collects linear algebraic criteria for testing positive semidefiniteness of Hermitian matrices. Lemma~\ref{le:srxsr} establishes the multiplicativity of Schmidt rank under tensor products, which will be used to analyse the rank structure of vectors in certain subspaces. 
Lemma~\ref{le:hyperplane} states that every hyperplane in a multipartite Hilbert space is spanned by product vectors. This property enables us to construct product vectors orthogonal to prescribed subspaces, a key step in deriving contradictions in the distillability analysis. 
Proposition~\ref{pro:1-distill} summarises the main results of~\cite{song2026entanglement} for a family of states $\rho$ diagonal in the five orthonormal pure states given in Eq.~\eqref{eq:rho}, with positive eigenvalues $\lambda_1,\ldots,\lambda_5$. 
It specifies precisely the parameter regions for which $\rho$ is PPT, NPT, or 1-distillable. However, the distillability of $\rho$ remains open when $\lambda_5$ lies in the interval $\left[\frac{24\sqrt{2}-33}{7},\frac{33-12\sqrt{6}}{25}\right)$. Example~\ref{ex:1/7} demonstrates that at $\lambda_5=\frac{1}{7}$, the state $\rho$ in Proposition~\ref{pro:1-distill} is 1-undistillable, thereby providing at least one explicit instance of such behaviour. 
Proposition~\ref{pro:1-distill} left the 1-distillability of $\rho$ undetermined for $\lambda_5$ in $\left[\frac{24\sqrt{2}-33}{7},\frac{33-12\sqrt{6}}{25}\right)$. We now resolve this by showing the states are 1-undistillable. For any rank-two projection $P$ on subsystem 
$A$, the operator $(P\otimes I_B)\rho^{\Gamma}(P^{\dagger}\otimes I_B)$ reduces to two forms. Applying Lemma~\ref{lem:principal_minors} shows that both forms are positive semidefinite for that interval. 
Proposition~\ref{pro:1-undistill} then establishes 1-undistillability for the symmetric rank-five NPT state $\rho$ of Proposition~\ref{pro:1-distill} on that interval. Together with Proposition~\ref{pro:1-distill} and Example~\ref{ex:1/7}, this settles the 1-distillability problem for the family.
For the same interval, we consider 2-distillability. The spectral decomposition of $\rho^{\Gamma}$ gives that of $(\rho^{\Gamma})^{\otimes 2}$. To prove 2-distillability, it suffices to find a Schmidt-rank-two vector $\ket{\psi}$ with $\bra{\psi}\sigma\ket{\psi}<0$, where $\sigma=(\rho^\Gamma)^{\otimes 2}$. We examine vectors in the negative subspace $\mathcal{N}(\sigma)$. A detailed analysis shows no such vector lies in $\mathcal{N}(\sigma) \oplus \operatorname{span}\{\ket{a_9,a_9}\}$. Proposition~\ref{pro:Ns} formalises this, stating that any Schmidt-rank-two $\ket{\psi}$ with $\bra{\psi}\sigma\ket{\psi}<0$ cannot be contained in that 17-dimensional subspace. We turn to a numerical investigation of the 2-distillability of $\rho$. 
We present the general form of a Schmidt-rank-two state $\ket{\psi}$, derive the condition for $\ket{\psi}$ to have Schmidt rank two, and obtain a compact expression for $\bra{\psi}\sigma\ket{\psi}$. This expression is represented as a quadratic form $\mathbf{d}^{\dagger}M\mathbf{d}$. Due to the complexity of this quadratic form, we decompose it into a sum of 16 lower-dimensional quadratic forms, $\sum_{i=1}^{16} \mathbf{v}_i^\dagger M_i \mathbf{v}_i$, where each $M_i$ has nonzero matrix elements.

The rest of this paper is organized as follows. 
In Sec.~\ref{sec:Pre}, we introduce the primary facts and knowledge used in this work. In Sec.~\ref{sec:result}, we present the main results. Finally, we conclude in Sec.~\ref{sec:conclusion}.
\section{Preliminaries}
\label{sec:Pre}
In this section, we introduce the main techniques and facts used in this paper. Let $\mathcal{H} = \mathcal{H}_{A}\otimes \mathcal{H}_{B}$ be the bipartite Hilbert space with $\dim \mathcal{H}_{A} = M$ and $\dim \mathcal{H}_{B} = N$. We study bipartite quantum states $\rho$ on $\mathcal{H}$. We denote the range and kernel of a linear map $\rho$ by $\mathcal{R}(\rho)$ and $\ker \rho$, respectively. 
Unless stated otherwise, the states will not be normalized. We denote orthonormal bases of $\mathcal{H}_{A}$ and $\mathcal{H}_{B}$ by $\{|i\rangle_{A}:i = 0,1,\dots ,M - 1\}$ and $\{|j\rangle_{B}:j = 0,1,\dots ,N - 1\}$, respectively. 

The distillability problem involves many-copy states from a composite system. 
Let $\rho_{A_iB_i}$ be an $M_i \times N_i$ state of rank $r_i$ acting on the Hilbert space $\mathcal{H}_{A_i} \otimes \mathcal{H}_{B_i}$, $i=1,2$. Consider a state $\rho$ on the composite system $A_1, B_1, A_2, B_2$, acting on $\mathcal{H}_{A_1} \otimes \mathcal{H}_{B_1} \otimes \mathcal{H}_{A_2} \otimes \mathcal{H}_{B_2}$. By interchanging the two middle tensor factors, we may regard $\rho$ as a bipartite state on $\mathcal{H}_A \otimes \mathcal{H}_B$, where $\mathcal{H}_A = \mathcal{H}_{A_1} \otimes \mathcal{H}_{A_2}$ and $\mathcal{H}_B = \mathcal{H}_{B_1} \otimes \mathcal{H}_{B_2}$. In this case we write $\rho = \rho_{A_1A_2:B_1B_2}$. Thus $\rho$ is an $M_1M_2 \times N_1N_2$ state whose rank does not exceed $r_1r_2$. In particular, for the tensor product $\rho = \rho_{A_1B_1} \otimes \rho_{A_2B_2}$, it is easy to verify that $\rho$ is an $M_1M_2 \times N_1N_2$ state of rank exactly $r_1r_2$. The above construction readily generalizes to the tensor product of $N$ states $\rho_{A_iB_i}$, $i=1,\ldots,N$, which yields a bipartite state on the Hilbert space $\mathcal{H}_{A_1,\ldots,A_N} \otimes \mathcal{H}_{B_1,\ldots,B_N}$. When each $\mathcal{H}_{A_i} \otimes \mathcal{H}_{B_i} = \mathcal{H}$, this space is denoted by $\mathcal{H}^{\otimes N}$. 
With these preparations, we are now in a position to define the distillability of entangled states.
\begin{definition}
\label{def:distill}
A bipartite state $\rho$ is $n$-distillable under LOCC if there exists a Schmidt-rank-two state $|\psi \rangle \in \mathcal{H}^{\otimes n}$ such that $\langle \psi |(\rho^{\otimes n})^{\Gamma}|\psi \rangle < 0$. Here, the state $|\psi \rangle$ can be written as the superposition of two pure product states. Equivalently, $\rho$ is $n$-distillable under LOCC if there exists a rank-two projection operator $P$ on subsystem $A^n$ such that the matrix $(P\otimes I_{B^n})(\rho^{\otimes n})^{\Gamma}(P^{\dagger}\otimes I_{B^n})$ has at least one negative eigenvalue. If a finite $n$ exists then we say that $\rho$ is distillable. If no such $n$ exists, $\rho$ is called undistillable.
\end{definition}
An entangled state $\rho$ that is not distillable is referred to as bound entangled. PPT entangled states constitute a prominent class of bound entangled states. The longstanding distillability problem asks whether there exist bound entangled states that are NPT. Next, we recall some fundamental results from linear algebra that are useful for determining whether a given matrix is positive semidefinite.
\begin{lemma}\label{lem:principal_minors}
    For a given $n \times n$ Hermitian matrix $A$, we have 
    \begin{enumerate}[label=(\roman*)]
        \item $A$ is positive semidefinite if and only if all principal minors of $A$ are nonnegative;
        \item If the first $n-1$ leading principal minors (respectively, the last $n-1$ trailing principal minors) of $A$ are positive and $\det A \geq 0$, then $A$ is positive semidefinite. 
    \end{enumerate}
\end{lemma}
In the study of the 2-distillability problem, one often needs to determine whether the Schmidt rank of an entangled state equals two. We now recall some basic facts concerning the Schmidt rank.
\begin{lemma}
\label{le:srxsr}
Let $\ket{x},\ket{y}\in \bbC^m \otimes \bbC^n$ have Schmidt rank $a$ and $b$, respectively. Then the bipartite vector $\ket{x}\otimes\ket{y}\in \bbC^{m^2} \otimes \bbC^{n^2}$ has Schmidt rank $ab$.
\end{lemma}
\begin{lemma}
\label{le:hyperplane}
Every hyperplane (i.e., a $(D-1)$-dimensional subspace) in the $n$-partite Hilbert space $\bbC^{d_1}\otimes...\otimes \bbC^{d_n}$ with $D=\prod^n_{j=1} d_j$ is spanned by product vectors.
\end{lemma}

The following fact is from~\cite{song2026entanglement}. 
\begin{proposition}\label{pro:1-distill}
Let $\rho$ be a two-qutrit symmetric NPT state of rank five, with an orthonormal basis $\{|e_{1}\rangle ,\ldots ,|e_{5}\rangle \}$ for $\mathcal{R}(\rho)$ given explicitly by
\begin{equation}
\label{eq:rho} 
\begin{array}{l}
|e_{1}\rangle=\frac{|01\rangle+|10\rangle}{\sqrt{2}},\quad|e_{2}\rangle=\frac{|12\rangle+|21\rangle}{\sqrt{2}},\quad|e_{3}\rangle=\frac{|02\rangle+|20\rangle}{\sqrt{2}},\\
|e_{4}\rangle=\frac{|00\rangle-|11\rangle}{\sqrt{2}},\quad|e_{5}\rangle=\frac{|00\rangle+|11\rangle-2|22\rangle}{\sqrt{6}}.
\end{array}
\end{equation}
Suppose $\rho$ has the spectral decomposition $\rho = \sum_{j = 1}^{5}\lambda_{j}|e_{j}\rangle \langle e_{j}|$, where $\lambda_{5} = x$ and $\lambda_{1} = \lambda_{2} = \lambda_{3} = \lambda_{4} = \frac{1 - x}{4}$. Then $\rho$ is NPT if and only if $x\in (0,\frac{33 - 12\sqrt{6}}{25})\cup (\frac{3}{11},1)$, and $\rho$ is PPT when $x \in [\frac{33-12\sqrt{6}}{25},\frac{3}{11}]$. Moreover, $x\in (0,\frac{24\sqrt{2} - 33}{7})\cup (\frac{3}{11},1)$ is a sufficient condition for $\rho$ to be 1-distillable.
\end{proposition}

\begin{example}\label{ex:1/7}
The state $\rho$ given in Proposition~\ref{pro:1-distill} is 1-undistillable if $\lambda_5 = x = \frac{1}{7}\in [\frac{24\sqrt{2} - 33}{7},\frac{33 - 12\sqrt{6}}{25})$. 
\end{example}

The state $\rho$ in Proposition~\ref{pro:1-distill} is NPT for $\lambda_5 \in [\frac{24\sqrt{2}-33}{7},\frac{33-12\sqrt{6}}{25}) \approx [0.134,0.144)$, but whether $\rho$ is 1-distillable in this interval remains unclear. Motivated by Example~\ref{ex:1/7}, we conjecture that $\rho$ is 1-undistillable for $\lambda_5 \in [\frac{24\sqrt{2} - 33}{7},\frac{33 - 12\sqrt{6}}{25}) \approx [0.134,0.144)$.

\section{Result}
\label{sec:result}
In this section, we study the 1-distillability and 2-distillability of the rank-five NPT states $\rho$ from Proposition~\ref{pro:1-distill}. Proposition~\ref{pro:1-undistill} establishes 1-undistillability for these states when $\lambda_5$ lies in the interval $[\frac{24\sqrt{2}-33}{7},\frac{33-12\sqrt{6}}{25})$. For the same interval, we proceed to analyze 2-distillability. Proposition~\ref{pro:Ns} shows that any Schmidt-rank-two vector $\ket{\psi}$ with $\bra{\psi}\sigma\ket{\psi}<0$ cannot belong to the 17-dimensional subspace $\mathcal{N}(\sigma) \oplus \operatorname{span}\{\ket{a_9,a_9}\}$, where $\sigma=(\rho^\Gamma)^{\otimes 2}$. We then present the general form of such a Schmidt-rank-two state and derive a compact expression for its expectation value $\bra{\psi}\sigma\ket{\psi}$. This expression takes the quadratic form $\mathbf{d}^{\dagger}M\mathbf{d}$. To handle its complexity, we decompose this quadratic form into a sum of 16 lower-dimensional quadratic forms.
\begin{proposition}\label{pro:1-undistill}
The state $\rho$ in Proposition~\ref{pro:1-distill} is 1-undistillable for $\lambda_5 \in [\frac{24\sqrt{2}-33}{7},\frac{33-12\sqrt{6}}{25}) \approx [0.134,0.144)$.
\end{proposition}

\begin{proof}
The projection $P$ on subsystem $A$ can be restricted to one of the two forms
\begin{equation}
\begin{array}{c}
     P_1 = \begin{pmatrix}
         1 & a & 0 \\
         0 & 0 & 1 \\
         0 & 0 & 0
     \end{pmatrix},\qquad
     P_2 = \begin{pmatrix}
         1 & 0 & b \\
         0 & 1 & c \\
         0 & 0 & 0
     \end{pmatrix}
\end{array}
\end{equation}
where $a,b,c \in \mathbb{C}$ are parameters to be chosen such that the projected matrix is not positive semidefinite. 
We define the resulting matrices
\begin{equation}
\label{eq:alpha}
    \alpha_i := (P_i \otimes I_3)\rho^{\Gamma}(P_i^{\dagger} \otimes I_3),\qquad i=1,2
\end{equation}

For $\alpha_1$, we apply Lemma~\ref{lem:principal_minors} (i) and examine all its principal minors. Detailed calculations are provided in Appendix~\ref{app:alpha_1}. All first-order principal minors are nonnegative. Given that some second-order principal minors factor as products of first-order principal minors, we only need to consider those with non-zero off-diagonal entries. Similarly, since some third-order principal minors are the product of a second-order principal minor and a first-order principal minor, we only consider the case where all off-diagonal entries are nonzero. All principal minors of order four or higher are products of lower-order principal minors. It is easy to show that all principal minors are nonnegative. Thus, $\alpha_1$ is positive semidefinite for all $a$.

For $\alpha_2$, we apply Lemma~\ref{lem:principal_minors} (ii) and examine its leading principal minors. See Appendix~\ref{app:alpha_2} for the detailed calculation. Through calculation for $x$ and by completing the square for some terms, we find that the first six leading principal minors are positive, while all higher-order leading principal minors are zero. Therefore, $\alpha_2$ is positive semidefinite for all $b,c$. Consequently, no choice of $a,b,c$ can render the projected matrix indefinite, which implies that $\rho$ is $1$-undistillable. 
\end{proof}
Next, we study the two-copy distillability of the state $\rho$ in Proposition~\ref{pro:1-distill}. Let the spectral decomposition be 
\begin{eqnarray}
\label{eq:rho^Gamma}
\rho^\Gamma=
\sum^9_{j=1}\lambda_j \proj{a_j},
\end{eqnarray}
where $\lambda_1\ge...\ge\lambda_8>0>\lambda_9$ and $\ket{a_9}$ has Schmidt rank three. The explicit expression is given in Appendix~\ref{sec:spec}. We have
\begin{eqnarray}
&&
\sigma:=(\rho^\Gamma)^{\otimes2}=  
\sum^9_{j=1}\lambda_j \ket{a_j}\bra{a_j}
\otimes
\sum^9_{k=1}\lambda_k \ket{a_k}\bra{a_k}
\notag\\
&& \in B(H_{A_1B_1}\otimes H_{A_2B_2})
:=B(H_{A_1A_2:B_1B_2})\cong B(\bbC^9\otimes \bbC^9).
\notag\\
\label{eq:sigma}
\end{eqnarray}
Hence $\sigma$ is a $9\times9$ Hermitian matrix of inertia $(16,0,65)$. The negative subspace $\cN(\sigma)$ of $\sigma$, being a $16$-dimensional subspace of $\bbC^9\otimes \bbC^9$, is spanned by $\ket{a_j,a_9}$ and $\ket{a_9,a_j}$ for $j=1,...,8$. Then $\rho$ is 2-distillable if there is a bipartite state of Schmidt rank two 
\begin{eqnarray}
\label{eq:ps in C9 otimes C9}    
\ket{\ps}\in H_{A_1A_2}\otimes H_{B_1B_2}= \bbC^9\otimes \bbC^9,
\end{eqnarray}
such that $\bra{\ps} \sigma \ket{\ps}<0$.

\begin{proposition}\label{pro:Ns}
Let $\rho$ be the 1-undistillable two-qutrit state in Proposition~\ref{pro:1-distill}, and let $\sigma$ be the corresponding operator in \eqref{eq:sigma}. 
If there exists a vector $\ket{\ps}$ of Schmidt rank two such that $\bra{\ps}\sigma\ket{\ps}<0$, then $\ket{\ps}$ cannot lie entirely in the $17$-dimensional bipartite subspace $\cN(\s) \oplus \lin\{\ket{a_9,a_9}\}\subset\bbC^9\otimes\bbC^9$.
\end{proposition}
\begin{proof}
    Suppose $\ket{\ps}=\sum^8_{j=1} \alpha_j\ket{a_j,a_9}+\sum^8_{j=1} \beta_j\ket{a_9,a_j}$
    for complex numbers $\alpha_j,\beta_j$. Because $\ket{\ps}$ has Schmidt rank two, by Lemma \ref{le:srxsr} and \eqref{eq:sigma} we see that one of $\alpha_j$'s is nonzero, and one of $\beta_j$'s is also nonzero. Using Lemma \ref{le:hyperplane}, one can find a product vector $\ket{y}\in H_{A_1B_1}$ orthogonal to $\ket{a_9}$, and non-orthogonal to $\sum^8_{j=1}\alpha_j\ket{a_j}$ at the same time. Hence $\bra{y}\ket{\ps}\propto\ket{a_9} \in H_{A_2B_2}$ is a bipartite vector of Schmidt rank at most two. This contradicts the fact stated below Eq.~\eqref{eq:rho^Gamma} that $\ket{a_9}$ has Schmidt rank three. By a similar argument, one can show that this $17$-dimensional subspace actually contains no Schmidt-rank-two vector at all. 
\end{proof}
The above observation implies that the existence of $\ket{\psi}$ in Eq.~\eqref{eq:ps in C9 otimes C9} is tied to the positive subspace $\mathcal{P}(\sigma)$ of $\sigma$. This subspace is 65-dimensional in $\mathbb{C}^9\otimes \mathbb{C}^9$ and is spanned by $\ket{a_j,a_k}$ for $j,k=1,\dots,8$ and $\ket{a_9,a_9}$.

We now turn to a numerical investigation of 2-distillability for $\rho$. Specifically, we construct explicit expansions of $\ket{\psi}$ and $\sigma$ using the orthonormal basis $\{\ket{a_i}\}_{i=1}^9$ and the spectral decomposition of $\rho^\Gamma$ given in Appendix~\ref{sec:spec}. For the two-qutrit state $\rho$ in Proposition~\ref{pro:1-distill}, $\sigma$ is a linear combination of symmetric and antisymmetric states with counts 45 and 36, respectively. According to Definition~\ref{def:distill}, $\rho$ is 2-distillable if there exists an entangled state $\ket{\psi}$ of Schmidt rank two such that $\bra{\psi}\sigma\ket{\psi}<0$. We therefore study the Schmidt rank of $\ket{\psi}$ with respect to the bipartition $A_1A_2 : B_1B_2$ and derive a compact expression for $\bra{\psi}\sigma\ket{\psi}$ in terms of the expansion coefficients and the eigenvalues of $\rho^\Gamma$. The general form of a Schmidt-rank-two (SR-2) state is given by 
\begin{align}
\label{eq:psi-SR2}
    \ket{\psi} = \ket{x_0}_{A_1A_2}\otimes\ket{x_1}_{B_1B_2} + \ket{x_2}_{A_1A_2}\otimes\ket{x_3}_{B_1B_2},
\end{align}
where the constituent states satisfy $\langle x_0|x_2\rangle = \langle x_1|x_3\rangle = 0$.
Expanding each $\ket{x_i}$ in the orthonormal computational basis $\{\ket{mn}\}_{m,n=0}^{2}$ for the two-qutrit subsystems, we define the complex coefficients
$d_{i,(m,n)} = \langle mn|x_i\rangle$, so that
\begin{align}
\label{eq:xi}
    \ket{x_i} = \sum_{m,n=0}^{2} d_{i,(m,n)}\ket{mn}, \qquad i=0,1,2,3.
\end{align}
The orthogonality conditions then become
\begin{align}
\label{eq:SR2condition}
    \sum_{m,n=0}^{2} d_{0,(m,n)}^*d_{2,(m,n)} = \sum_{m,n=0}^{2} d_{1,(m,n)}^*d_{3,(m,n)} = 0.
\end{align}
We introduce a $9\times 9$ matrix $D$ with entries
\begin{align}
    D_{(m,p),(n,q)} &= d_{0,(m,p)}d_{1,(n,q)} + d_{2,(m,p)}d_{3,(n,q)}.
\end{align}
Its vectorization is denoted by
$\mathbf{d} = \operatorname{vec}(D) \in \mathbb{C}^{81}$. 
This vector serves as the primary coefficient representation of the state in the two bases used below.

For the quantity $\bra{\psi}\sigma\ket{\psi}$, we start from its quadratic-form representation $\bra{\psi}\sigma\ket{\psi}=\mathbf{d}^{\dagger}M\mathbf{d}$, with $M$ derived in Appendix~\ref{app:M}. 
We then decompose this quadratic form into a sum of 16 lower-dimensional quadratic forms. The index sets $T_1,\dots,T_{16}$ are defined in Appendix~\ref{app:Mi}. For each $i=1,\dots,16$, let $\mathbf{v}_i$ be the subvector of $\mathbf{d}$ consisting of the entries whose indices lie in $T_i$, and let $M_i$ be the principal submatrix of $M$ indexed by $T_i$. The explicit expression of $M_i$ is given in Appendix~\ref{app:Mi}. The original quadratic form then decomposes as
\begin{eqnarray}
    \bra{\psi}\sigma\ket{\psi} = \sum_{i=1}^{16} \mathbf{v}_i^\dagger M_i \mathbf{v}_i .
\end{eqnarray}
The goal is to show the non-negativity of the above expression under the orthogonality condition in Eq.~\eqref{eq:SR2condition}. First, finding a constructive method for $\psi$ that satisfies this orthogonality condition is nontrivial. Moreover, the quadratic-form representation $\bra{\ps}\sigma\ket{\ps}=\mathbf{d}^{\dagger}M\mathbf{d}$ is too complicated to allow a direct verification of non-negativity for all SR-2 states. 

\section{Conclusions}
\label{sec:conclusion}
We have completely resolved the 1-distillability of the rank-five symmetric two-qutrit NPT entangled family, proving that the previously open interval is 1-undistillable via the principal minors criterion. 
Over the same interval, we have also characterized a structural obstruction to 2-distillability, showing that no Schmidt-rank-two vector with a negative expectation value lies in this 17-dimensional subspace. 
Additionally, we have carried out explicit numerical investigations to further explore the 2-distillability of $\rho$.
Key contributions of this work include the complete classification of 1-distillability, a narrowed search space for 2-distillability, and a demonstration of the utility of linear algebraic positivity tests for high-dimensional entanglement distillation. 
A natural next step is to resolve the 2-distillability problem for this rank-five state. Another interesting direction is to construct more two-qutrit symmetric states and investigate their distillability.
\section*{Acknowledgments}
\addcontentsline{toc}{section}{Acknowledgments}
The first three authors were supported by the NNSF of China(Grant No. 12471427).

\addcontentsline{toc}{section}{References}
\bibliographystyle{unsrt}     
\phantomsection 
\bibliography{Result/main}

\newpage

\appendix
\renewcommand{\theequation}{\thesection-\arabic{equation}}
\section{Detailed calculations of principal minors for \texorpdfstring{$\alpha_1$}{alpha1}} 
\label{app:alpha_1}
\setcounter{equation}{0}   

In the proof of Proposition~\ref{pro:1-undistill}, we claimed that all principal minors of $\alpha_1$ are nonnegative. Here we present the explicit computations. 
After projecting, we obtain a $9 \times 9$ matrix $\alpha_1$ in Eq.~\eqref{eq:alpha_1}. 
The 1-undistillability of $\rho$ can be judged by testing the positive semidefiniteness of $\alpha_1$ and $\alpha_2$. 

Due to the presence of zero entries, we only present the submatrix containing the nonzero elements, as shown in Eq.~\eqref{eq:alpha_1}.
The subsequent computation of principal minors of various orders is also carried out based solely on Eq.~\eqref{eq:alpha_1}.
\begin{equation}
\label{eq:alpha_1}
\renewcommand\arraystretch{1.2}
\begin{pmatrix}
\frac{3+x}{24} + \frac{1-x}{8}|a|^2 & \frac{7x-3}{24}a + \frac{1-x}{8}\bar{a} & 0 & 0 & 0 & \frac{1-x}{8} \\[4pt]
\frac{7x-3}{24}\bar{a} + \frac{1-x}{8}a & \frac{1-x}{8} + \frac{3+x}{24}|a|^2 & 0 & 0 & 0 & \frac{1-x}{8}a \\[4pt]
0 & 0 & \frac{1-x}{8} + \frac{1-x}{8}|a|^2 & -\frac{x}{3} & -\frac{x}{3}a & 0 \\[4pt]
0 & 0 & -\frac{x}{3} & \frac{1-x}{8} & 0 & 0 \\[4pt]
0 & 0 & -\frac{x}{3}\bar{a} & 0 & \frac{1-x}{8} & 0 \\[4pt]
\frac{1-x}{8} & \frac{1-x}{8}\bar{a} & 0 & 0 & 0 & \frac{2x}{3}
\end{pmatrix}.
\end{equation} 

The following are the second-order principal minors with nonzero off-diagonal entries.
\begin{align}
\det \alpha_1[1,2] &= \frac{1}{24^2}\bigl[3(3+x)(1-x)|a|^4 + 48x(1-x)|a|^2 \nonumber\\
&\qquad + 3(3-7x)(1-x)(a^2 + \bar{a}^2) + 3(3+x)(1-x)\bigr], \tag{A1} \\ 
\det \alpha_1[1,6] &= \frac{1}{24^2}\bigl[48x(1-x)|a|^2 + 16x(3+x) - 9(1-x)^2\bigr], \tag{A2} \\
\det \alpha_1[2,6] &= \frac{1}{24^2}\bigl[(7x^2+66x-9)|a|^2 + 48x(1-x)\bigr], \tag{A3} \\
\det \alpha_1[3,4] &= \frac{1}{24^2}\bigl[9(1-x)^2|a|^2 + (3-7x)(3+x)\bigr], \tag{A4} \\
\det \alpha_1[3,5] &= \frac{1}{24^2}\bigl[(-55x^2-18x+9)|a|^2 + 9(1-x)^2\bigr]. \tag{A5}
\end{align} 

The third-order principal minors with nonzero off-diagonal entries are as follows.
\begin{align}
\det\alpha_1[1,2,6] &= \frac{1}{24^3}\bigl[(7x^2+66x-9)|a|^4 + 8x(35x-3)|a|^2 \nonumber\\
&\qquad + (7x-3)(3-19x)(a^2+\bar{a}^2) + 7x^2+66x-9\bigr], \tag{A6} \\
\det\alpha_1[3,4,5] &= \frac{3}{24^3}(x-1)(11x-3)(5x+3)(1+|a|^2). \tag{A7}
\end{align} 

All principal minors of order four or higher can be expressed as products of lower-order principal minors. Consequently, all principal minors are nonnegative.
\section{Detailed calculation of leading principal minors for \texorpdfstring{$\alpha_2$}{alpha2}}
\label{app:alpha_2}
\setcounter{equation}{0}   
In the proof of Proposition~\ref{pro:1-undistill}, we claimed that all leading principal minors of $\alpha_2$ are nonnegative. 
Here we present the explicit computations. 
After projecting, we obtain a $9 \times 9$ matrix $\alpha_2$. Due to the presence of zero entries, we only present the submatrix containing the nonzero elements, as shown in Eq.~\eqref{eq:alpha_2}. 
All subsequent calculations of principal minors are based solely on this reduced form.  
Let $b = b_1 + b_2i$ and $c = c_1 + c_2i$, and let $D_k$ denote the $k$-th leading principal minor of $\alpha_2$. Notably, all leading principal minors of $\alpha_2$ of order greater than six are zero. 
Accordingly, we present the analysis of the first four leading principal minors in Appendix~\ref{app:alpha2.D1-4}, the fifth-order leading principal minors in Appendix~\ref{app:alpha_2.D5}, and the sixth-order leading principal minors in Appendix~\ref{app:alpha_2.D6}.

\begin{equation}
\label{eq:alpha_2}
\renewcommand\arraystretch{2.0} 
\setlength\arraycolsep{3pt}     
\resizebox{0.9\linewidth}{!}   
{$
\begin{pmatrix}
\frac{3+x}{24} + \frac{1-x}{8}|b|^2 & 0 & -\frac{x}{3}b + \frac{1-x}{8}\bar{b} & \frac{1-x}{8}b\bar{c} & \frac{1-x}{8} & \frac{1-x}{8}\bar{c} \\
0 & \frac{1-x}{8}(1+|b|^2) & 0 & \frac{7x-3}{24} & \frac{1-x}{8}b\bar{c} & -\frac{x}{3}b \\
-\frac{x}{3}\bar{b} + \frac{1-x}{8}b & 0 & \frac{1-x}{8} + \frac{2x}{3}|b|^2 & -\frac{x}{3}\bar{c} & \frac{1-x}{8}b & \frac{2x}{3}b\bar{c} \\
\frac{1-x}{8}\bar{b}c & \frac{7x-3}{24} & -\frac{x}{3}c & \frac{1-x}{8} (1+|c|^2) & 0 & 0 \\
\frac{1-x}{8} & \frac{1-x}{8}\bar{b}c & \frac{1-x}{8}\bar{b} & 0 & \frac{3+x}{24}+\frac{1-x}{8}|c|^2 & -\frac{x}{3}c+\frac{1-x}{8}\bar{c} \\
\frac{1-x}{8}c & -\frac{x}{3}\bar{b} & \frac{2x}{3}\bar{b}c & 0 & -\frac{x}{3}\bar{c}+\frac{1-x}{8}c & \frac{1-x}{8}+\frac{2x}{3}|c|^2
\end{pmatrix}.$} 
\end{equation}

\subsection{Analysis of the first four leading principal minors}
\label{app:alpha2.D1-4}
According to Eq.~\eqref{eq:alpha_2}, the following are the first two leading principal minors. 
\begin{align}
D_1 &= \frac{1}{24}\bigl[3+x + 3(1-x)|b|^2\bigr], \\ 
D_2 &= \frac{3}{24^2}\bigl[3+x + 3(1-x)|b|^2\bigr](1-x)(1+|b|^2). 
\end{align}

It is easy to verify that for $x \in \left[\frac{24\sqrt{2}-33}{7},\ \frac{33-12\sqrt{6}}{25}\right]$, the first- and second-order leading principal minors are all positive. 
After substituting $b = b_1 + b_2i$ and $c = c_1 + c_2i$ into the third- and fourth-order leading principal minors and sorting the resulting terms in descending order first by total power, then by power of $b$, then by power of $b_1$, then by power of $c$, and finally by power of $c_1$, we obtain 
\begin{align}
    D_3 
    &=\frac{1}{24^3}\biggl[144x(x-1)^2|b|^6 + 144x(x-1)^2|b|^4 + 288x(x-1)^2|b|^2b_1^2  \nonumber \\
    & \quad + 27(x-1)^2(11x+1)b_1^2 + 9(x-1)^2(x+3)b_2^2 + 9(x-1)^2(x+3)\biggr], \\ 
    D_4 
    &= \frac{1}{24^4}\biggl[432x(1-x)^3|b|^6 + 9(x-1)^2(7x^2+66x-9)|b|^4|c|^2 \nonumber \\
    & \quad +96x(x - 1)(11x^2 + 6x - 9)|b|^2b_1^2 + 384x^2(x - 1)(5x - 3)|b|^2b_2^2 \nonumber \\
    & \quad + 12x(x-1)(19x^2 + 90x -45)|b|^2|c|^2 \nonumber \\
    & \quad + 3(x - 1)(1271x^3 - 777x^2 + 45x - 27)b_1^2 -27(x - 1)^3(x + 3)b_2^2 \nonumber\\
    & \quad + 31(x-1)(x+3)(5x+3)(11x-3)|c|^2 \nonumber \\
    & \quad + 24x(x-1)(5x-3)(x+3)\biggr]. 
\end{align}
From the above mathematical expressions, since the coefficients of the parameterized terms are nonnegative on the interval $\left[ \frac{24\sqrt{2} - 33}{7}, \frac{33 - 12\sqrt{6}}{25} \right]$ and the constant term is positive, the third- and fourth-order leading principal minors are positive. 
\subsection{Analysis of the fifth-order leading principal minor}
\label{app:alpha_2.D5}
Substitute $b = b_1 + b_2i$ and $c = c_1 + c_2i$ into the fifth-order leading principal minors, and sort the resulting terms in descending order first by total power, then by power of $b$, then by power of $b_1$, then by power of $c$, and finally by power of $c_1$, as shown in Eq.~\eqref{eq:D5}. 
For the fifth-order leading principal minor $D_5$, since some terms become negative for $x \in \left[ \frac{24\sqrt{2} - 33}{7}, \frac{33 - 12\sqrt{6}}{25} \right]$, it is necessary to complete the squares for certain terms, along with a case analysis based on the values of $b$ and $c$. 
\begin{align}\label{eq:D5}
    D_5 &= \frac{1}{24^5}[-27(x - 1)^3(7x^2 + 66x - 9)(b_1^6+b_2^6) \nonumber \\
    & \quad -81(x - 1)^3(7x^2 + 66x - 9)(b_1^4b_2^2+b_1^2b_2^4) \nonumber \\
    & \quad -54(x - 1)^3(7x^2 + 66x - 9)(b_1^4c_1^2+b_2^4c_2^2) \nonumber \\
    & \quad -162(x - 1)^3(91x^2 - 30x + 3)(b_1^4c_2^2+b_2^4c_1^2) \nonumber \\
    & \quad -432x(x - 1)^3(35x - 3)(b_1^2b_2^2c_1^2+b_1^2b_2^2c_2^2) \nonumber \\
    & \quad + 216(x - 1)^3(7x - 3)(19x - 3)(b_1^3b_2c_1c_2 + b_1b_2^3c_1c_2) \nonumber \\
    & \quad -27(x - 1)^3(7x^2 + 66x - 9)(b_1^2c_1^4+b_1^2c_2^4+b_2^2c_1^4+b_2^2c_2^4) \nonumber \\
    & \quad -54(x - 1)^3(7x^2 + 66x - 9)(b_1^2c_1^2c_2^2+b_2^2c_1^2c_2^2) \nonumber \\
    & \quad -3(x - 1)(1061x^4 - 5424x^3 + 3906x^2 - 648x + 81)b_1^4 \nonumber \\
    & \quad + 3(x - 1)(91x^4 + 3120x^3 - 2754x^2 + 648x - 81)b_2^4 \nonumber \\
    & \quad -6(x - 1)(485x^4 - 4272x^3 + 3330x^2 - 648x + 81)b_1^2b_2^2 \nonumber \\
    & \quad -6(x - 1)(613x^4 - 2520x^3 + 1746x^2 - 432x + 81)b_1^2c_1^2 \nonumber \\
    & \quad -6(x - 1)(37x^4 - 1368x^3 + 1170x^2 - 432x + 81)b_2^2c_2^2 \nonumber \\
    & \quad -6(x - 1)(4267x^4 - 6984x^3 + 1854x^2 + 432x - 81)b_1^2c_2^2 \nonumber \\
    & \quad -6(x - 1)(1003x^4 - 1992x^3 + 126x^2 + 432x - 81)b_2^2c_1^2 \nonumber \\
    & \quad +72(x - 1)^2(5x + 3)(7x - 3)(11x - 3)b_1b_2c_1c_2 \nonumber \\
    & \quad -9(x-1)^2(x+3)(5x+3)(11x-3)c_1^4 \nonumber \\
    & \quad -18(x-1)^2(x+3)(5x+3)(11x-3)c_1^2c_2^2 \nonumber \\
    & \quad -9(x-1)^2(x+3)(5x+3)(11x-3)c_2^4 \nonumber \\
    & \quad +48x(x - 1)(377x^3 - 321x^2 + 99x - 27)b_1^2 \nonumber \\
    & \quad + 144x(x - 1)^2(19x^2 - 24x + 9)b_2^2 \nonumber \\
    & \quad -48x(x-1)(35x^3-63x^2-63x+27)(c_1^2+c_2^2) \nonumber \\
    & \quad -192x^2(x-3)(x-1)(5x-3)]. 
\end{align}
We next present the analysis of $D_5$. 
For $x \in \left[ \frac{24\sqrt{2} - 33}{7}, \frac{33 - 12\sqrt{6}}{25} \right]$, the coefficients of all terms in $D_5$ are even and nonnegative, except for those involving $b_1^3b_2c_1c_2$, $b_1b_2^3c_1c_2$, $b_1^2c_2^2$, $b_2^2c_1^2$, and $b_1b_2c_1c_2$. Since the constant term is positive, it suffices to show that these exceptional terms are nonnegative, which proves that $D_5$ is positive. 
We analyze the terms $b_1^2c_2^2$ and $b_2^2c_1^2$ in Section~\ref{sec:D5.nonpositive}. The remaining three terms, $b_1^3b_2c_1c_2$, $b_1b_2^3c_1c_2$, and $b_1b_2c_1c_2$, are related and are discussed in Section~\ref{sec:D5.related}.

\subsubsection{Analysis of nonpositive terms in \texorpdfstring{$D_5$}{D5}}\label{sec:D5.nonpositive}
The terms involving $b_1^2c_2^2$ and $b_2^2c_1^2$ are nonpositive whenever their coefficients are nonpositive. For these terms, we complete the squares by selecting appropriate components, ensuring that the coefficients of the unknowns in the completed-square form are nonnegative for $x \in \left[ \frac{24\sqrt{2} - 33}{7}, \frac{33 - 12\sqrt{6}}{25} \right]$, thereby proving that $D_5$ is positive. The details of the square completion are presented in Eq.~\eqref{eq:D5.b1^2c2^2} for $b_1^2c_2^2$ and in Eq.~\eqref{eq:D5.b2^2c1^2} for $b_2^2c_1^2$. Due to the complexity of the expression, we show only the terms necessary for completing the squares and omit irrelevant components. 

For the term $b_1^2c_2^2$, its coefficient is nonpositive only when $x \in \left[ \frac{24\sqrt{2} - 33}{7}, x_0 \right]$, with $x_0 < \frac{33 - 12\sqrt{6}}{25}$. 
Consequently, the square completion shown in Eq.~\eqref{eq:D5.b1^2c2^2} applies solely on this subinterval. In contrast, for the term $b_2^2c_1^2$, the coefficient remains nonpositive throughout the entire interval $x \in \left[ \frac{24\sqrt{2} - 33}{7}, \frac{33 - 12\sqrt{6}}{25} \right]$. Therefore, Eq.~\eqref{eq:D5.b2^2c1^2} is valid on the whole interval.
\begin{align}\label{eq:D5.b1^2c2^2}
    & \quad -3(x - 1)(1061x^4 - 5424x^3 + 3906x^2 - 648x + 81)b_1^4 \nonumber \\
    & \quad -9(x-1)^2(x+3)(5x+3)(11x-3)c_2^4 \nonumber \\
    & \quad -6(x - 1)(4267x^4 - 6984x^3 + 1854x^2 + 432x - 81)b_1^2c_2^2 \nonumber \\
    & = 3(x - 1)(4267x^4 - 6984x^3 + 1854x^2 + 432x - 81)(b_1^4 + c_2^4 - 2b_1^2c_2^2) \nonumber \\
    & \quad -72x(x - 1)(222x^3 - 517x^2 + 240x - 9)b_1^4 \nonumber \\
    & \quad -24x(x - 1)(554x^3 - 825x^2 + 180x + 27)c_2^4 \nonumber \\
    & = 3(x - 1)(4267x^4 - 6984x^3 + 1854x^2 + 432x - 81)(b_1^2 - c_2^2)^2 \nonumber \\ 
    & \quad -72x(x - 1)(222x^3 - 517x^2 + 240x - 9)b_1^4 \nonumber \\
    & \quad -24x(x - 1)(554x^3 - 825x^2 + 180x + 27)c_2^4 \nonumber \\
    & \geq 0 . 
\end{align}
\begin{align}\label{eq:D5.b2^2c1^2}
    & \quad 3(x - 1)(91x^4 + 3120x^3 - 2754x^2 + 648x - 81)b_2^4 \nonumber \\
    & \quad -9(x-1)^2(x+3)(5x+3)(11x-3)c_1^4 \nonumber \\
    & \quad -6(x - 1)(1003x^4 - 1992x^3 + 126x^2 + 432x - 81)b_2^2c_1^2 \nonumber \\
    & = 3(x - 1)(1003x^4 - 1992x^3 + 126x^2 + 432x - 81)(b_2^4 + c_1^4 - 2b_2^2c_1^2) \nonumber \\
    & \quad -72x(x - 1)(38x^3 - 213x^2 + 120x - 9)b_2^4 \nonumber \\
    & \quad -24x(x - 1)(146x^3 - 201x^2 - 36x + 27)c_1^4 \nonumber \\
    & = 3(x - 1)(1003x^4 - 1992x^3 + 126x^2 + 432x - 81)(b_2^2 - c_1^2)^2 \nonumber \\ 
    & \quad -72x(x - 1)(38x^3 - 213x^2 + 120x - 9)b_2^4 \nonumber \\
    & \quad -24x(x - 1)(146x^3 - 201x^2 - 36x + 27)c_1^4 \nonumber \\
    & \geq 0 .
\end{align}

\subsubsection{Analysis of interrelated terms in \texorpdfstring{$D_5$}{D5}}\label{sec:D5.related}
The terms $b_1^3b_2c_1c_2$, $b_1b_2^3c_1c_2$, and $b_1b_2c_1c_2$ share a common factor $b_1b_2c_1c_2$ and are therefore related. 
When $b_1b_2c_1c_2 = 0$, the other two terms also vanish. 
When $b_1b_2c_1c_2$ is negative, the terms $b_1^3b_2c_1c_2$ and $b_1b_2^3c_1c_2$ are positive. 
For $b_1b_2c_1c_2$ itself, we complete squares for selected terms to prove that $D_5$ is positive, as shown in Eq.~\eqref{eq:D5.b1b2c1c2}.\\
\begin{align}\label{eq:D5.b1b2c1c2}
    & \quad -6(x - 1)(485x^4 - 4272x^3 + 3330x^2 - 648x + 81)b_1^2b_2^2 \nonumber \\
    & \quad -18(x-1)^2(x+3)(5x+3)(11x-3)c_1^2c_2^2 \nonumber \\
    & \quad -6(x - 1)(613x^4 - 2520x^3 + 1746x^2 - 432x + 81)b_1^2c_1^2 \nonumber \\
    & \quad -6(x - 1)(37x^4 - 1368x^3 + 1170x^2 - 432x + 81)b_2^2c_2^2 \nonumber \\
    & \quad +72(x - 1)^2(5x + 3)(7x - 3)(11x - 3)b_1b_2c_1c_2 \nonumber \\
    & = 18(x - 1)^2(5x + 3)(7x - 3)(11x - 3)(b_1^2b_2^2 + c_1^2c_2^2 + 2b_1b_2c_1c_2) \nonumber \\
    & \quad + 18(x - 1)^2(5x + 3)(7x - 3)(11x - 3)(b_1^2c_1^2 + b_2^2c_2^2 + 2b_1b_2c_1c_2) \nonumber \\
    & \quad -48x(x - 1)(5x - 3)(41x^2 - 114x + 9)b_1^2b_2^2 \nonumber \\
    & \quad -144x(x - 1)^2(5x + 3)(11x - 3)c_1^2c_2^2 \nonumber \\
    & \quad -48x^2(x - 1)(13x - 21)(17x - 9)b_1^2c_1^2 \nonumber \\
    & \quad -48x^2(x - 1)(149x^2 - 330x + 117)b_2^2c_2^2 \nonumber \\
    & = 18(x - 1)^2(5x + 3)(7x - 3)(11x - 3)(b_1b_2 + c_1c_2)^2 \nonumber \\
    & \quad + 18(x - 1)^2(5x + 3)(7x - 3)(11x - 3)(b_1c_1 + b_2c_2)^2 \nonumber \\ 
    & \quad -48x(x - 1)(5x - 3)(41x^2 - 114x + 9)b_1^2b_2^2 \nonumber \\
    & \quad -144x(x - 1)^2(5x + 3)(11x - 3)c_1^2c_2^2 \nonumber \\
    & \quad -48x^2(x - 1)(13x - 21)(17x - 9)b_1^2c_1^2 \nonumber \\
    & \quad -48x^2(x - 1)(149x^2 - 330x + 117)b_2^2c_2^2 \nonumber \\
    & \geq 0 .
\end{align}
When the term $b_1b_2c_1c_2$ is positive, the terms involving $b_1^3b_2c_1c_2$ and $b_1b_2^3c_1c_2$ are all negative. For these negative terms, we complete squares for some terms to prove that $D_5$ is positive, as shown in Eq.~\eqref{eq:D5.b_1^3 b_2 c_1 c_2 + b_1 b_2^3 c_1 c_2}. \\
\begin{align}\label{eq:D5.b_1^3 b_2 c_1 c_2 + b_1 b_2^3 c_1 c_2}
    & \quad -162(x - 1)^3(91 x^2 - 30 x + 3)(b_1^4 c_2^2+b_2^4 c_1^2) \nonumber \\
    & \quad -432 x(x - 1)^3(35 x - 3)(b_1^2 b_2^2 c_1^2+b_1^2 b_2^2 c_2^2) \nonumber \\
    & \quad +216(x - 1)^3(7 x - 3)(19 x - 3)(b_1^3 b_2 c_1 c_2 + b_1 b_2^3 c_1 c_2) \nonumber \\
    & = -108(x - 1)^3(7x - 3)(19x - 3)(b_1^4c_2^2+b_2^4c_1^2 + b_1^2b_2^2c_1^2+b_1^2b_2^2c_2^2 \nonumber \\ 
    & \quad \quad-2b_1^3b_2c_1c_2 - 2b_1b_2^3c_1c_2) \nonumber \\
    & \quad -54(x - 1)^3(7x^2 + 66x - 9)(b_1^4 c_2^2+b_2^4 c_1^2) \nonumber \\
    & \quad -108(x - 1)^3(7x^2 + 66x - 9)(b_1^2 b_2^2 c_1^2+b_1^2 b_2^2 c_2^2) \nonumber \\
    & = -108(x - 1)^3(7x - 3)(19x - 3)|b|^2(b_1c_2 - b_2c_1)^2 \nonumber \\
    & \quad -54(x - 1)^3(7x^2 + 66x - 9)(b_1^4 c_2^2+b_2^4 c_1^2) \nonumber \\
    & \quad -108(x - 1)^3(7x^2 + 66x - 9)(b_1^2 b_2^2 c_1^2+b_1^2 b_2^2 c_2^2) \nonumber \\
    & \geq 0 . 
\end{align}
After completing the squares, for $x \in \left[ \frac{24\sqrt{2} - 33}{7}, \frac{33 - 12\sqrt{6}}{25} \right]$, all variable terms are nonnegative and the constant term is positive. Hence $D_5$ is positive. 
\subsection{Analysis of the sixth-order leading principal minor}
\label{app:alpha_2.D6}
Substitute $b = b_1 + b_2i$ and $c = c_1 + c_2i$ into the sixth-order leading principal minors, and sort the resulting terms in descending order first by total power, then by power of $b$, then by power of $b_1$, then by power of $c$, and finally by power of $c_1$, as shown in Eq.~\eqref{eq:D6}. 
For the sixth-order leading principal minor $D_6$, since some terms become negative for $x \in \left[ \frac{24\sqrt{2} - 33}{7}, \frac{33 - 12\sqrt{6}}{25} \right]$, it is necessary to complete the squares for certain terms, along with a case analysis based on the values of $b$ and $c$. 
\begin{align}\label{eq:D6}
D_6 &= \frac{1}{24^6}\biggl[-27 (x - 1)^{2} (5x + 3) (11x - 3) (7x^{2} + 66x - 9)(b_1^6 + 3b_1^4b_2^2 + 3b_1^2b_2^4 + b_2^6 ) \nonumber \\
& \quad -27 (x - 1)^{2} (5x + 3) (11x - 3) (7x^{2} + 66x - 9)(b_1^4c_1^2 + b_2^4c_2^2) \nonumber \\
& \quad -9(x-1)^2(5x+3)(11x-3)(553x^2-114x+9)(b_1^4c_2^2+b_2^4c_1^2) \nonumber \\
& \quad +72(x - 1)^{2} (5x + 3) (7x - 3) (11x - 3) (19x - 3)(b_1^3b_2c_1c_2 + b_1b_2^3c_1c_2) \nonumber \\
& \quad -72(x-1)^2(5x+3)(11x-3)(287x^2+42x-9)(b_1^2b_2^2c_1^2+b_1^2b_2^2c_2^2) \nonumber \\
& \quad -27(x - 1)^2(5x + 3)(11x - 3)(7x^2 + 66x - 9)(b_1^2c_1^4+b_2^2c_2^4) \nonumber \\
& \quad -18(x - 1)^2(5x + 3)(11x - 3)(287x^2 + 42x - 9)(b_1^2c_1^2c_2^2 + b_2^2c_1^2c_2^2) \nonumber \\
& \quad -9(x - 1)^2(5x + 3)(11x - 3)(553x^2 - 114x + 9)(b_1^2c_2^4 + b_2^2c_1^4) \nonumber \\
& \quad +72(x - 1)^2(5x + 3)(7x - 3)(11x - 3)(19x - 3)(b_1b_2c_1^3c_2 + b_1b_2c_1c_2^3) \nonumber \\
& \quad -27 (x - 1)^{2} (5x + 3) (11x - 3) (7x^{2} + 66x - 9)( c_1^6 + 3c_1^4c_2^2 + 3c_1^2c_2^4 + c_2^6) \nonumber \\
& \quad -27(x - 1)^2(2825x^4 + 1168x^3 - 1110x^2 + 216x - 27)b_1^4 \nonumber \\
& \quad -9(x - 1)^2(1435x^4 + 1200x^3 - 2178x^2 + 648x - 81)b_2^4 \nonumber \\
& \quad -18(x - 1)^2(4955x^4 + 2352x^3 - 2754x^2 + 648x - 81)b_1^2b_2^2 \nonumber \\
& \quad -54(x - 1)^2(2825x^4 + 1168x^3 - 1110x^2 + 216x - 27)b_1^2c_1^2 \nonumber \\
& \quad +18(x - 1)^2(13355x^4 - 11808x^3 + 3438x^2 + 216x - 81)(b_1^2c_2^2 + b_2^2c_1^2) \nonumber \\
& \quad -72(x - 1)^2(9155x^4 - 4728x^3 + 342x^2 + 432x - 81)b_1b_2c_1c_2 \nonumber \\
& \quad -18(x - 1)^2(1435x^4 + 1200x^3 - 2178x^2 + 648x - 81)b_2^2c_2^2 \nonumber \\
& \quad -27(x - 1)^2(2825x^4 + 1168x^3 - 1110x^2 + 216x - 27)c_1^4 \nonumber \\
& \quad -18(x - 1)^2(4955x^4 + 2352x^3 - 2754x^2 + 648x - 81)c_1^2c_2^2 \nonumber \\
& \quad -9(x - 1)^2(1435x^4 + 1200x^3 - 2178x^2 + 648x - 81)c_2^4 \nonumber \\  
& \quad -432x(x - 1)^2(115x^3 - 75x^2 + 33x - 9)b_1^2 \nonumber \\
& \quad -144x(x - 1)^2(25x^3 - 33x^2 + 99x - 27)b_2^2 \nonumber \\
& \quad -432x(x - 1)^2(115x^3 - 75x^2 + 33x - 9)c_1^2 \nonumber \\
& \quad -144x(x - 1)^2(25x^3 - 33x^2 + 99x - 27)c_2^2 \nonumber \\
& \quad +576x^2(x-3)(x-1)^2(5x-3)\biggr].
\end{align} 
We now analyze $D_6$ as follows. First, since the coefficients of the unknowns are all even and nonnegative for $x \in \left[ \frac{24\sqrt{2} - 33}{7}, \frac{33 - 12\sqrt{6}}{25} \right]$, all terms are nonnegative except those involving $b_1^3b_2c_1c_2$, $b_1b_2^3c_1c_2$, $b_1b_2c_1^3c_2$, $b_1b_2c_1c_2^3$, $b_1^2c_2^2$, $b_2^2c_1^2$, and $b_1b_2c_1c_2$. 
Since the constant term is positive, it suffices to prove that the terms involving $b_1^3b_2c_1c_2$, $b_1b_2^3c_1c_2$, $b_1b_2c_1^3c_2$, $b_1b_2c_1c_2^3$, $b_1^2c_2^2$, $b_2^2c_1^2$, and $b_1b_2c_1c_2$ are nonnegative to conclude that $D_6$ is positive. 
We analyze the terms $b_1^2c_2^2$ and $b_2^2c_1^2$ in Section~\ref{sec:D6.nonpositive}. The remaining five terms, $b_1^3b_2c_1c_2$, $b_1b_2^3c_1c_2$, $b_1b_2c_1^3c_2$, $b_1b_2c_1c_2^3$, and $b_1b_2c_1c_2$, are related and are discussed in Section~\ref{sec:D6.related}.

\subsubsection{Analysis of nonpositive terms in \texorpdfstring{$D_6$}{D6}}\label{sec:D6.nonpositive}

The terms involving $b_1^2c_2^2$ and $b_2^2c_1^2$ are always nonpositive. For these nonpositive terms, we complete squares for some terms to prove that $D_6$ is positive, as shown in Eq.~\eqref{eq:D6.b1^2c2^2} and Eq.~\eqref{eq:D6.b2^2c1^2}.  Due to the length of the expression, we only present the terms required for completing the squares, omitting irrelevant parts.\\
\begin{align}\label{eq:D6.b1^2c2^2}
    & \quad -27(x - 1)^2(2825x^4 + 1168x^3 - 1110x^2 + 216x - 27)b_1^4 \nonumber \\
    & \quad -9(x - 1)^2(1435x^4 + 1200x^3 - 2178x^2 + 648x - 81)c_2^4 \nonumber \\
    & \quad +18(x - 1)^2(13355x^4 - 11808x^3 + 3438x^2 + 216x - 81)b_1^2c_2^2 \nonumber \\
    & = -9(x - 1)^2(13355x^4 - 11808x^3 + 3438x^2 + 216x - 81)(b_1^4 + c_2^4 - 2b_1^2c_2^2) \nonumber \\
    & \quad +144x(x - 1)^2(305x^3 - 957x^2 + 423x - 27)b_1^4 \nonumber \\
    & \quad +144x(x - 1)^2(745x^3 - 813x^2 + 351x - 27)c_2^4 \nonumber \\
    & = -9(x - 1)^2(13355x^4 - 11808x^3 + 3438x^2 + 216x - 81)(b_1^2 - c_2^2)^2 \nonumber \\ 
    & \quad +144x(x - 1)^2(305x^3 - 957x^2 + 423x - 27)b_1^4 \nonumber \\
    & \quad +144x(x - 1)^2(745x^3 - 813x^2 + 351x - 27)c_2^4 \nonumber \\
    & \geq 0 . 
\end{align}
\begin{align}\label{eq:D6.b2^2c1^2}
    & \quad -9(x - 1)^2(1435x^4 + 1200x^3 - 2178x^2 + 648x - 81)b_2^4 \nonumber \\
    & \quad -27(x - 1)^2(2825x^4 + 1168x^3 - 1110x^2 + 216x - 27)c_1^4 \nonumber \\
    & \quad +18(x - 1)^2(13355x^4 - 11808x^3 + 3438x^2 + 216x - 81)b_2^2c_1^2 \nonumber \\
    & = -9(x - 1)^2(13355x^4 - 11808x^3 + 3438x^2 + 216x - 81)(b_2^4 + c_1^4 - 2b_2^2c_1^2) \nonumber \\
    & \quad +144x(x - 1)^2(745x^3 - 813x^2 + 351x - 27)b_2^4 \nonumber \\
    & \quad +144x(x - 1)^2(305x^3 - 957x^2 + 423x - 27)c_1^4 \nonumber \\
    & = -9(x - 1)^2(13355x^4 - 11808x^3 + 3438x^2 + 216x - 81)(b_2^2 - c_1^2)^2 \nonumber \\ 
    & \quad +144x(x - 1)^2(745x^3 - 813x^2 + 351x - 27)b_2^4 \nonumber \\
    & \quad +144x(x - 1)^2(305x^3 - 957x^2 + 423x - 27)c_1^4 \nonumber \\
    & \geq 0 . 
\end{align}

\subsubsection{Analysis of interrelated terms in \texorpdfstring{$D_6$}{D6}}\label{sec:D6.related}
The terms $b_1^3b_2c_1c_2$, $b_1b_2^3c_1c_2$, $b_1b_2c_1^3c_2$ and $b_1b_2c_1c_2^3$ and $b_1b_2c_1c_2$ share a common factor $b_1b_2c_1c_2$ and are therefore related. 
When $b_1b_2c_1c_2 = 0$, all the terms involving $b_1^3b_2c_1c_2$, $b_1b_2^3c_1c_2$, $b_1b_2c_1^3c_2$ and $b_1b_2c_1c_2^3$ vanish. 
When $b_1b_2c_1c_2$ is negative, those terms are all positive. We complete squares for some terms to prove that $D_6$ is positive, as shown in Eq.~\eqref{eq:D6.b1b2c1c2}.\\
\begin{align}\label{eq:D6.b1b2c1c2}
    & \quad -18(x - 1)^2(4955x^4 + 2352x^3 - 2754x^2 + 648x - 81)(b_1^2b_2^2 + c_1^2c_2^2) \nonumber \\
    & \quad -54(x - 1)^2(2825x^4 + 1168x^3 - 1110x^2 + 216x - 27)b_1^2c_1^2 \nonumber \\
    & \quad -18(x - 1)^2(1435x^4 + 1200x^3 - 2178x^2 + 648x - 81)b_2^2c_2^2 \nonumber \\
    & \quad -72(x - 1)^2(9155x^4 - 4728x^3 + 342x^2 + 432x - 81)b_1b_2c_1c_2 \nonumber \\
    & = -18(x - 1)^2(9155x^4 - 4728x^3 + 342x^2 + 432x - 81)(b_1^2b_2^2 + c_1^2c_2^2 + 2b_1b_2c_1c_2) \nonumber \\
    & \quad -18(x - 1)^2(9155x^4 - 4728x^3 + 342x^2 + 432x - 81)(b_1^2c_1^2 + b_2^2c_2^2 + 2b_1b_2c_1c_2) \nonumber \\
    & \quad + 432x(x - 1)^3(5x - 3)(35x - 3)(b_1^2b_2^2 + c_1^2c_2^2) \nonumber \\
    & \quad  + 144x(x - 1)^2(85x^3 - 1029x^2 + 459x - 27)b_1^2c_1^2 \nonumber \\
    & \quad + 144x(x - 1)^2(965x^3 - 741x^2 + 315x - 27)b_2^2c_2^2 \nonumber \\
    & = -18(x - 1)^2(9155x^4 - 4728x^3 + 342x^2 + 432x - 81)(b_1b_2 + c_1c_2)^2 \nonumber \\ 
    & \quad -18(x - 1)^2(9155x^4 - 4728x^3 + 342x^2 + 432x - 81)(b_1c_1 + b_2c_2)^2 \nonumber \\
    & \quad + 432x(x - 1)^3(5x - 3)(35x - 3)(b_1^2b_2^2 + c_1^2c_2^2) \nonumber \\
    & \quad  + 144x(x - 1)^2(85x^3 - 1029x^2 + 459x - 27)b_1^2c_1^2 \nonumber \\
    & \quad + 144x(x - 1)^2(965x^3 - 741x^2 + 315x - 27)b_2^2c_2^2 \nonumber \\
    & \geq 0 .
\end{align}
When $b_1b_2c_1c_2$ is positive, the terms $b_1^3b_2c_1c_2$, $b_1b_2^3c_1c_2$, $b_1b_2c_1^3c_2$, and $b_1b_2c_1c_2^3$ are all negative. To handle these negative terms, we complete squares for selected terms to prove that $D_6$ is positive. The details of the square completion are given in Eq.~\eqref{eq:D5.b1^2c2^2} for $b_1^3b_2c_1c_2$ and $b_1b_2^3c_1c_2$, and in Eq.~\eqref{eq:D5.b2^2c1^2} for $b_1b_2c_1^3c_2$ and $b_1b_2c_1c_2^3$.
\begin{align}
    & \quad -9(x-1)^2(5x+3)(11x-3)(553x^2-114x+9)(b_1^4c_2^2+b_2^4c_1^2) \nonumber \\
    & \quad -72(x-1)^2(5x+3)(11x-3)(287x^2+42x-9)(b_1^2b_2^2c_1^2+b_1^2b_2^2c_2^2) \nonumber \\
    & \quad +72(x - 1)^{2} (5x + 3) (7x - 3) (11x - 3) (19x - 3)(b_1^3b_2c_1c_2 + b_1b_2^3c_1c_2) \nonumber \\
    & = -36(x - 1)^{2} (5x + 3) (7x - 3) (11x - 3) (19x - 3)(b_1^4c_2^2+b_2^4c_1^2 + b_1^2b_2^2c_1^2+b_1^2b_2^2c_2^2 \nonumber \\
    & \quad \quad - 2b_1^3b_2c_1c_2 - 2b_1b_2^3c_1c_2) \nonumber \\
    & \quad -27(x - 1)^2(5x + 3)(11x - 3)(7x^2 + 66x - 9)(b_1^4c_2^2+b_2^4c_1^2) \nonumber \\
    & \quad -324(x - 1)^2(5x + 3)(11x - 3)(49x^2 + 18x - 3)(b_1^2b_2^2c_1^2+b_1^2b_2^2c_2^2) \nonumber \\
    & = -36(x - 1)^{2} (5x + 3) (7x - 3) (11x - 3) (19x - 3)|b|^2(b_1c_2 - b_2c_1)^2 \nonumber \\
    & \quad -27(x - 1)^2(5x + 3)(11x - 3)(7x^2 + 66x - 9)(b_1^4c_2^2+b_2^4c_1^2) \nonumber \\
    & \quad -324(x - 1)^2(5x + 3)(11x - 3)(49x^2 + 18x - 3)(b_1^2b_2^2c_1^2+b_1^2b_2^2c_2^2) \nonumber \\
    & \geq 0 . 
\end{align}
\begin{align}
    & \quad -9(x - 1)^2(5x + 3)(11x - 3)(553x^2 - 114x + 9)(b_1^2c_2^4 + b_2^2c_1^4) \nonumber \\
    & \quad -18(x - 1)^2(5x + 3)(11x - 3)(287x^2 + 42x - 9)(b_1^2c_1^2c_2^2 + b_2^2c_1^2c_2^2) \nonumber \\
    & \quad +72(x - 1)^2(5x + 3)(7x - 3)(11x - 3)(19x - 3)(b_1b_2c_1^3c_2 + b_1b_2c_1c_2^3) \nonumber \\
    & = -36(x - 1)^2(5x + 3)(7x - 3)(11x - 3)(19x - 3)(b_1^2c_2^4 + b_2^2c_1^4 + b_1^2c_1^2c_2^2 + b_2^2c_1^2c_2^2 \nonumber \\
    & \quad \quad -2b_1b_2c_1^3c_2 - 2b_1b_2c_1c_2^3) \nonumber \\
    & \quad -27(x - 1)^2(5x + 3)(11x - 3)(7x^2 + 66x - 9)(b_1^2c_2^4 + b_2^2c_1^4) \nonumber \\
    & \quad -54(x - 1)^2(5x + 3)(11x - 3)(7x^2 + 66x - 9)(b_1^2c_1^2c_2^2 + b_2^2c_1^2c_2^2) \nonumber \\
    & = -36(x - 1)^2(5x + 3)(7x - 3)(11x - 3)(19x - 3)|c|^2(b_1c_2-b_2c_1)^2 \nonumber \\
    & \quad -27(x - 1)^2(5x + 3)(11x - 3)(7x^2 + 66x - 9)(b_1^2c_2^4 + b_2^2c_1^4) \nonumber \\
    & \quad -54(x - 1)^2(5x + 3)(11x - 3)(7x^2 + 66x - 9)(b_1^2c_1^2c_2^2 + b_2^2c_1^2c_2^2) \nonumber \\
    & \geq 0 . 
\end{align}

After completing the squares, for $x \in \left[ \frac{24\sqrt{2} - 33}{7}, \frac{33 - 12\sqrt{6}}{25} \right]$, all variable terms are nonnegative and the constant term is positive. From the above analysis, we conclude that $D_6$ is positive. 
Consequently, the first six leading principal minors are positive.
\section{The spectral decomposition and base of \texorpdfstring{$\rho^{\Gamma}$}{rho^{\Gamma}}}
\label{sec:spec}
This section presents the spectral decomposition of the partially transposed density matrix $\rho^{\Gamma}$, together with the explicit expressions for its eigenvalues and eigenvectors. Based on this decomposition and the bipartition defined in Eq.~\eqref{eq:sigma}, we construct the basis states for the symmetric and antisymmetric subspaces $\mathcal{N}(\sigma)$ and $\mathcal{P}(\sigma)$. The spectral decomposition and the constructed bases are of fundamental importance for the subsequent analysis.

The spectral decomposition of $\rho^{\Gamma}$ is given by $\rho^{\Gamma} = \sum_{i=1}^{9} \lambda_i \ket{a_i}\bra{a_i}$, where the eigenvalues $\lambda_i$ and the corresponding eigenvectors $\ket{a_i}$ are specified as follows.
\begin{eqnarray}
\label{eq:spectral of rho^Gamma}
    &&
    \lambda_1 = \dfrac{3+7x + 3\sqrt{11x^2-10x+3}}{24},\notag\\
    &&
    |a_1\rangle \propto (1-x)\bigl(|00\rangle+|11\rangle\bigr) - \bigl(1-3x-\sqrt{11x^2-10x+3}\bigr)|22\rangle,\notag\\[4pt]
    &&
    \lambda_2 = \dfrac{3-5x}{12},\quad
    |a_2\rangle = \dfrac{|01\rangle - |10\rangle}{\sqrt{2}},\notag\\[4pt]
    &&
    \lambda_3 = \lambda_4 = \dfrac{3+5x}{24},\quad
    |a_3\rangle = \dfrac{|02\rangle - |20\rangle}{\sqrt{2}},\quad
    |a_4\rangle = \dfrac{|12\rangle - |21\rangle}{\sqrt{2}},\notag\\[4pt]
    &&
    \lambda_5 = \lambda_6 = \dfrac{3-11x}{24},\quad
    |a_5\rangle = \dfrac{|02\rangle + |20\rangle}{\sqrt{2}},\quad
    |a_6\rangle = \dfrac{|12\rangle + |21\rangle}{\sqrt{2}},\notag\\[4pt]
    &&
    \lambda_7 = \lambda_8 = \frac{x}{6},\quad
    |a_7\rangle = \dfrac{|00\rangle - |11\rangle}{\sqrt{2}},\quad
    |a_8\rangle = \dfrac{|01\rangle + |10\rangle}{\sqrt{2}},\notag\\[4pt]
    &&
    \lambda_9 = \dfrac{3+7x - 3\sqrt{11x^2-10x+3}}{24},\notag\\
    &&
    |a_9\rangle \propto (1-x)\bigl(|00\rangle+|11\rangle\bigr) - \bigl(1-3x+\sqrt{11x^2-10x+3}\bigr)|22\rangle. 
\notag\\
\end{eqnarray}
The normalization constants for the unnormalized states $\ket{a_1}$ and $\ket{a_9}$ are respectively
$
N_\pm = \sqrt{2(1-x)^2 + \bigl(1-3x \mp \sqrt{11x^2-10x+3}\bigr)^2}
$. 
For $x\in \bigl[\frac{24\sqrt{2} - 33}{7},\frac{33 - 12\sqrt{6}}{25}\bigr]$, the eigenvalues obey the strict ordering 
\[
\lambda_1 > \lambda_2 > \lambda_3 = \lambda_4 > \lambda_5 = \lambda_6 > \lambda_7 = \lambda_8 > 3\abs{\lambda_9} > \abs{\lambda_9} > 0 > \lambda_9.
\]
According to the bipartition defined in Eq.~\eqref{eq:sigma} and Eq.~\eqref{eq:spectral of rho^Gamma}, $\sigma$ is a linear combination of symmetric and antisymmetric states. 
The symmetric part consists of tensor products of two symmetric components of $\rho^{\Gamma}$ or two antisymmetric components of $\rho^{\Gamma}$. The antisymmetric part consists of tensor products of one symmetric component and one antisymmetric component of $\rho^{\Gamma}$. 

The bases of the subspaces $\mathcal{N}(\sigma)$ and $\mathcal{P}(\sigma)$ are given as follows. For $\mathcal{N}(\sigma)$, the symmetric basis states are $\ket{a_9 a_i}$ and $\ket{a_i a_9}$ with $i=1,5,6,7,8$, while the antisymmetric basis states are the same forms with $i=2,3,4$. For $\mathcal{P}(\sigma)$, the symmetric basis states consist of $\ket{a_i a_j}$ for $i,j\in\{1,5,6,7,8\}$ and for $i,j\in\{2,3,4\}$, together with $\ket{a_9 a_9}$, the antisymmetric basis states are $\ket{a_i a_j}$ with $i\in\{1,5,6,7,8\}$ and $j\in\{2,3,4\}$, and with $i$ and $j$ interchanged.

\section{Two-copy entanglement distillation}
\label{app:two-copy distill}
\subsection{The deduction of matrix $M$}
\label{app:M}
The operator $\sigma$, derived from the spectral decomposition of $\rho^{\Gamma}$, is written as
\begin{align}
\label{eq:sigma_expanded_SR2}
\sigma 
    &= \sum_{i,j=1}^{9} \lambda_i\lambda_j \ket{a_i a_j}\bra{a_i a_j}_{A_1B_1A_2B_2} \nonumber\\
    &= \sum_{i,j=1}^{9} \lambda_i\lambda_j \ket{b_{ij}}\bra{b_{ij}}_{A_1A_2B_1B_2},
\end{align} 
where $\lambda_i$ are the eigenvalues and $\ket{a_i}$ the corresponding eigenvectors. Let $(s_i)_{mn} = \langle mn|a_i\rangle$, expanding the eigenvectors in the computational basis gives
\begin{align}
\label{eq:base_SR2}
    \ket{a_i} = \sum_{m,n=0}^{2} (s_i)_{mn}\ket{m}_{A_1}\ket{n}_{B_1},
    \quad
    \ket{a_j} = \sum_{p,q=0}^{2} (s_j)_{pq}\ket{p}_{A_2}\ket{q}_{B_2},
\end{align} 
and the tensor-product states are given by
$\ket{b_{ij}} = \sum_{m,n,p,q} (s_i)_{mn}(s_j)_{pq} \ket{mp}_{A_1A_2}\otimes\ket{nq}_{B_1B_2}$. 
The two forms correspond to the same operator expressed in two different subsystem orderings, yielding two equivalent expectation values. 

We now compute $\langle\psi|\sigma|\psi\rangle$ using the coefficient vector $\mathbf{d} = \operatorname{vec}(D)$ from the basis $A_1A_2 \otimes B_1B_2$. Define vectors $y_{ij}\in\mathbb{C}^{81}$ with components
$(y_{ij})_{mpnq} = (s_i)_{mn}(s_j)_{pq}$, so that $\langle b_{ij}|\psi\rangle = y_{ij}^{\dagger}\mathbf{d}$. 
With the matrix
\begin{align}
\label{eq:M=sum ij=1,...,9}
    M = \sum_{i,j=1}^{9} \lambda_i\lambda_j \, y_{ij} y_{ij}^{\dagger}\in\mathbb{C}^{81\times81},
\end{align} 
the expectation value becomes
\begin{align}
\label{eq:quadratic form1}
    \bra{\psi}\sigma\ket{\psi}
    &= \sum_{i,j=1}^{9} \lambda_i\lambda_j \left| \langle b_{ij}|\psi\rangle \right|^2 \nonumber\\
    &= \sum_{i,j=1}^{9} \lambda_i\lambda_j \, \mathbf{d}^{\dagger} y_{ij} y_{ij}^{\dagger} \mathbf{d}
    = \mathbf{d}^{\dagger} M \mathbf{d}.
\end{align} 
\subsection{The explicit expression of $M_i$}
\label{app:Mi}
The index sets $T_1,\dots,T_{16}$ are defined as follows.
\begin{equation}
    \begin{aligned}
    T_1 &= \{1,11,21,31,41,51,61,71,81\}, &
    T_2 &= \{2,10,32,40,62,70\}, \\
    T_3 &= \{3,19,33,49,63,79\}, &
    T_4 &= \{4,14,24,28,38,48\}, \\
    T_5 &= \{7,17,27,55,65,75\}, &
    T_6 &= \{12,20,42,50,72,80\}, \\
    T_7 &= \{34,44,54,58,68,78\}, &
    T_8 &= \{5,13,29,37\}, \\
    T_9 &= \{6,22,30,46\}, &
    T_{10} &= \{8,16,56,64\}, \\
    T_{11} &= \{9,25,57,73\}, &
    T_{12} &= \{15,23,39,47\}, \\
    T_{13} &= \{18,26,66,74\}, &
    T_{14} &= \{35,43,59,67\}, \\
    T_{15} &= \{36,52,60,76\}, &
    T_{16} &= \{45,53,69,77\}.
\end{aligned} 
\end{equation}
The matrix $M_1$ is defined as the principal submatrix of $M$ obtained by taking the rows and columns whose indices lie in $T_1$.
\begin{equation}
\resizebox{\textwidth}{!}{$
\begin{pmatrix}
  \frac{x^{2}}{576} + \frac{x}{96} + \frac{1}{64} & - \frac{x^{2}}{192} - \frac{x}{96} + \frac{1}{64} & - \frac{x^{2}}{192} - \frac{x}{96} + \frac{1}{64} & - \frac{x^{2}}{192} - \frac{x}{96} + \frac{1}{64} & \frac{x^{2}}{64} - \frac{x}{32} + \frac{1}{64} & \frac{x^{2}}{64} - \frac{x}{32} + \frac{1}{64} & - \frac{x^{2}}{192} - \frac{x}{96} + \frac{1}{64} & \frac{x^{2}}{64} - \frac{x}{32} + \frac{1}{64} & \frac{x^{2}}{64} - \frac{x}{32} + \frac{1}{64}\\
  - \frac{x^{2}}{192} - \frac{x}{96} + \frac{1}{64} & \frac{x^{2}}{576} + \frac{x}{96} + \frac{1}{64} & - \frac{x^{2}}{192} - \frac{x}{96} + \frac{1}{64} & \frac{x^{2}}{64} - \frac{x}{32} + \frac{1}{64} & - \frac{x^{2}}{192} - \frac{x}{96} + \frac{1}{64} & \frac{x^{2}}{64} - \frac{x}{32} + \frac{1}{64} & \frac{x^{2}}{64} - \frac{x}{32} + \frac{1}{64} & - \frac{x^{2}}{192} - \frac{x}{96} + \frac{1}{64} & \frac{x^{2}}{64} - \frac{x}{32} + \frac{1}{64}\\
  - \frac{x^{2}}{192} - \frac{x}{96} + \frac{1}{64} & - \frac{x^{2}}{192} - \frac{x}{96} + \frac{1}{64} & \frac{x \left(x + 3\right)}{36} & \frac{x^{2}}{64} - \frac{x}{32} + \frac{1}{64} & \frac{x^{2}}{64} - \frac{x}{32} + \frac{1}{64} & \frac{x \left(1 - x\right)}{12} & \frac{x^{2}}{64} - \frac{x}{32} + \frac{1}{64} & \frac{x^{2}}{64} - \frac{x}{32} + \frac{1}{64} & \frac{x \left(1 - x\right)}{12}\\
  - \frac{x^{2}}{192} - \frac{x}{96} + \frac{1}{64} & \frac{x^{2}}{64} - \frac{x}{32} + \frac{1}{64} & \frac{x^{2}}{64} - \frac{x}{32} + \frac{1}{64} & \frac{x^{2}}{576} + \frac{x}{96} + \frac{1}{64} & - \frac{x^{2}}{192} - \frac{x}{96} + \frac{1}{64} & - \frac{x^{2}}{192} - \frac{x}{96} + \frac{1}{64} & - \frac{x^{2}}{192} - \frac{x}{96} + \frac{1}{64} & \frac{x^{2}}{64} - \frac{x}{32} + \frac{1}{64} & \frac{x^{2}}{64} - \frac{x}{32} + \frac{1}{64}\\
  \frac{x^{2}}{64} - \frac{x}{32} + \frac{1}{64} & - \frac{x^{2}}{192} - \frac{x}{96} + \frac{1}{64} & \frac{x^{2}}{64} - \frac{x}{32} + \frac{1}{64} & - \frac{x^{2}}{192} - \frac{x}{96} + \frac{1}{64} & \frac{x^{2}}{576} + \frac{x}{96} + \frac{1}{64} & - \frac{x^{2}}{192} - \frac{x}{96} + \frac{1}{64} & \frac{x^{2}}{64} - \frac{x}{32} + \frac{1}{64} & - \frac{x^{2}}{192} - \frac{x}{96} + \frac{1}{64} & \frac{x^{2}}{64} - \frac{x}{32} + \frac{1}{64}\\
  \frac{x^{2}}{64} - \frac{x}{32} + \frac{1}{64} & \frac{x^{2}}{64} - \frac{x}{32} + \frac{1}{64} & \frac{x \left(1 - x\right)}{12} & - \frac{x^{2}}{192} - \frac{x}{96} + \frac{1}{64} & - \frac{x^{2}}{192} - \frac{x}{96} + \frac{1}{64} & \frac{x \left(x + 3\right)}{36} & \frac{x^{2}}{64} - \frac{x}{32} + \frac{1}{64} & \frac{x^{2}}{64} - \frac{x}{32} + \frac{1}{64} & \frac{x \left(1 - x\right)}{12}\\
  - \frac{x^{2}}{192} - \frac{x}{96} + \frac{1}{64} & \frac{x^{2}}{64} - \frac{x}{32} + \frac{1}{64} & \frac{x^{2}}{64} - \frac{x}{32} + \frac{1}{64} & - \frac{x^{2}}{192} - \frac{x}{96} + \frac{1}{64} & \frac{x^{2}}{64} - \frac{x}{32} + \frac{1}{64} & \frac{x^{2}}{64} - \frac{x}{32} + \frac{1}{64} & \frac{x \left(x + 3\right)}{36} & \frac{x \left(1 - x\right)}{12} & \frac{x \left(1 - x\right)}{12}\\
  \frac{x^{2}}{64} - \frac{x}{32} + \frac{1}{64} & - \frac{x^{2}}{192} - \frac{x}{96} + \frac{1}{64} & \frac{x^{2}}{64} - \frac{x}{32} + \frac{1}{64} & \frac{x^{2}}{64} - \frac{x}{32} + \frac{1}{64} & - \frac{x^{2}}{192} - \frac{x}{96} + \frac{1}{64} & \frac{x^{2}}{64} - \frac{x}{32} + \frac{1}{64} & \frac{x \left(1 - x\right)}{12} & \frac{x \left(x + 3\right)}{36} & \frac{x \left(1 - x\right)}{12}\\
  \frac{x^{2}}{64} - \frac{x}{32} + \frac{1}{64} & \frac{x^{2}}{64} - \frac{x}{32} + \frac{1}{64} & \frac{x \left(1 - x\right)}{12} & \frac{x^{2}}{64} - \frac{x}{32} + \frac{1}{64} & \frac{x^{2}}{64} - \frac{x}{32} + \frac{1}{64} & \frac{x \left(1 - x\right)}{12} & \frac{x \left(1 - x\right)}{12} & \frac{x \left(1 - x\right)}{12} & \frac{4 x^{2}}{9}
\end{pmatrix}.
$}\\
\end{equation}
The matrix $M_2$ is defined as the principal submatrix of $M$ obtained by taking the rows and columns whose indices lie in $T_2$.
\begin{equation}
\resizebox{\textwidth}{!}{$
\begin{pmatrix}
- \frac{x^{2}}{192} - \frac{x}{96} + \frac{1}{64} & \frac{7 x^{2}}{576} + \frac{x}{32} - \frac{1}{64} & \frac{x^{2}}{64} - \frac{x}{32} + \frac{1}{64} & - \frac{7 x^{2}}{192} + \frac{5 x}{96} - \frac{1}{64} & \frac{x^{2}}{64} - \frac{x}{32} + \frac{1}{64} & - \frac{7 x^{2}}{192} + \frac{5 x}{96} - \frac{1}{64}\\
\frac{7 x^{2}}{576} + \frac{x}{32} - \frac{1}{64} & - \frac{x^{2}}{192} - \frac{x}{96} + \frac{1}{64} & - \frac{7 x^{2}}{192} + \frac{5 x}{96} - \frac{1}{64} & \frac{x^{2}}{64} - \frac{x}{32} + \frac{1}{64} & - \frac{7 x^{2}}{192} + \frac{5 x}{96} - \frac{1}{64} & \frac{x^{2}}{64} - \frac{x}{32} + \frac{1}{64}\\
\frac{x^{2}}{64} - \frac{x}{32} + \frac{1}{64} & - \frac{7 x^{2}}{192} + \frac{5 x}{96} - \frac{1}{64} & - \frac{x^{2}}{192} - \frac{x}{96} + \frac{1}{64} & \frac{7 x^{2}}{576} + \frac{x}{32} - \frac{1}{64} & \frac{x^{2}}{64} - \frac{x}{32} + \frac{1}{64} & - \frac{7 x^{2}}{192} + \frac{5 x}{96} - \frac{1}{64}\\
- \frac{7 x^{2}}{192} + \frac{5 x}{96} - \frac{1}{64} & \frac{x^{2}}{64} - \frac{x}{32} + \frac{1}{64} & \frac{7 x^{2}}{576} + \frac{x}{32} - \frac{1}{64} & - \frac{x^{2}}{192} - \frac{x}{96} + \frac{1}{64} & - \frac{7 x^{2}}{192} + \frac{5 x}{96} - \frac{1}{64} & \frac{x^{2}}{64} - \frac{x}{32} + \frac{1}{64}\\
\frac{x^{2}}{64} - \frac{x}{32} + \frac{1}{64} & - \frac{7 x^{2}}{192} + \frac{5 x}{96} - \frac{1}{64} & \frac{x^{2}}{64} - \frac{x}{32} + \frac{1}{64} & - \frac{7 x^{2}}{192} + \frac{5 x}{96} - \frac{1}{64} & \frac{x \left(1 - x\right)}{12} & \frac{x \left(7 x - 3\right)}{36}\\
- \frac{7 x^{2}}{192} + \frac{5 x}{96} - \frac{1}{64} & \frac{x^{2}}{64} - \frac{x}{32} + \frac{1}{64} & - \frac{7 x^{2}}{192} + \frac{5 x}{96} - \frac{1}{64} & \frac{x^{2}}{64} - \frac{x}{32} + \frac{1}{64} & \frac{x \left(7 x - 3\right)}{36} & \frac{x \left(1 - x\right)}{12}
\end{pmatrix}.
$}
\end{equation}
The matrix $M_3$ is defined as the principal submatrix of $M$ obtained by taking the rows and columns whose indices lie in $T_3$.
\begin{equation}
\resizebox{\textwidth}{!}{$
\begin{pmatrix}
- \frac{x^{2}}{192} - \frac{x}{96} + \frac{1}{64} & \frac{x \left(- x - 3\right)}{72} & \frac{x^{2}}{64} - \frac{x}{32} + \frac{1}{64} & \frac{x \left(x - 1\right)}{24} & \frac{x^{2}}{64} - \frac{x}{32} + \frac{1}{64} & \frac{x \left(x - 1\right)}{24}\\
\frac{x \left(- x - 3\right)}{72} & - \frac{x^{2}}{192} - \frac{x}{96} + \frac{1}{64} & \frac{x \left(x - 1\right)}{24} & \frac{x^{2}}{64} - \frac{x}{32} + \frac{1}{64} & \frac{x \left(x - 1\right)}{24} & \frac{x^{2}}{64} - \frac{x}{32} + \frac{1}{64}\\
\frac{x^{2}}{64} - \frac{x}{32} + \frac{1}{64} & \frac{x \left(x - 1\right)}{24} & - \frac{x^{2}}{192} - \frac{x}{96} + \frac{1}{64} & \frac{x \left(- x - 3\right)}{72} & \frac{x^{2}}{64} - \frac{x}{32} + \frac{1}{64} & \frac{x \left(x - 1\right)}{24}\\
\frac{x \left(x - 1\right)}{24} & \frac{x^{2}}{64} - \frac{x}{32} + \frac{1}{64} & \frac{x \left(- x - 3\right)}{72} & - \frac{x^{2}}{192} - \frac{x}{96} + \frac{1}{64} & \frac{x \left(x - 1\right)}{24} & \frac{x^{2}}{64} - \frac{x}{32} + \frac{1}{64}\\
\frac{x^{2}}{64} - \frac{x}{32} + \frac{1}{64} & \frac{x \left(x - 1\right)}{24} & \frac{x^{2}}{64} - \frac{x}{32} + \frac{1}{64} & \frac{x \left(x - 1\right)}{24} & \frac{x \left(1 - x\right)}{12} & - \frac{2 x^{2}}{9}\\
\frac{x \left(x - 1\right)}{24} & \frac{x^{2}}{64} - \frac{x}{32} + \frac{1}{64} & \frac{x \left(x - 1\right)}{24} & \frac{x^{2}}{64} - \frac{x}{32} + \frac{1}{64} & - \frac{2 x^{2}}{9} & \frac{x \left(1 - x\right)}{12}
\end{pmatrix}.
$}
\end{equation}
The matrix $M_4$ is defined as the principal submatrix of $M$ obtained by taking the rows and columns whose indices lie in $T_4$.
\begin{equation}
\resizebox{\textwidth}{!}{$
\begin{pmatrix}
- \frac{x^{2}}{192} - \frac{x}{96} + \frac{1}{64} & \frac{x^{2}}{64} - \frac{x}{32} + \frac{1}{64} & \frac{x^{2}}{64} - \frac{x}{32} + \frac{1}{64} & \frac{7 x^{2}}{576} + \frac{x}{32} - \frac{1}{64} & - \frac{7 x^{2}}{192} + \frac{5 x}{96} - \frac{1}{64} & - \frac{7 x^{2}}{192} + \frac{5 x}{96} - \frac{1}{64}\\
\frac{x^{2}}{64} - \frac{x}{32} + \frac{1}{64} & - \frac{x^{2}}{192} - \frac{x}{96} + \frac{1}{64} & \frac{x^{2}}{64} - \frac{x}{32} + \frac{1}{64} & - \frac{7 x^{2}}{192} + \frac{5 x}{96} - \frac{1}{64} & \frac{7 x^{2}}{576} + \frac{x}{32} - \frac{1}{64} & - \frac{7 x^{2}}{192} + \frac{5 x}{96} - \frac{1}{64}\\
\frac{x^{2}}{64} - \frac{x}{32} + \frac{1}{64} & \frac{x^{2}}{64} - \frac{x}{32} + \frac{1}{64} & \frac{x \left(1 - x\right)}{12} & - \frac{7 x^{2}}{192} + \frac{5 x}{96} - \frac{1}{64} & - \frac{7 x^{2}}{192} + \frac{5 x}{96} - \frac{1}{64} & \frac{x \left(7 x - 3\right)}{36}\\
\frac{7 x^{2}}{576} + \frac{x}{32} - \frac{1}{64} & - \frac{7 x^{2}}{192} + \frac{5 x}{96} - \frac{1}{64} & - \frac{7 x^{2}}{192} + \frac{5 x}{96} - \frac{1}{64} & - \frac{x^{2}}{192} - \frac{x}{96} + \frac{1}{64} & \frac{x^{2}}{64} - \frac{x}{32} + \frac{1}{64} & \frac{x^{2}}{64} - \frac{x}{32} + \frac{1}{64}\\
- \frac{7 x^{2}}{192} + \frac{5 x}{96} - \frac{1}{64} & \frac{7 x^{2}}{576} + \frac{x}{32} - \frac{1}{64} & - \frac{7 x^{2}}{192} + \frac{5 x}{96} - \frac{1}{64} & \frac{x^{2}}{64} - \frac{x}{32} + \frac{1}{64} & - \frac{x^{2}}{192} - \frac{x}{96} + \frac{1}{64} & \frac{x^{2}}{64} - \frac{x}{32} + \frac{1}{64}\\
- \frac{7 x^{2}}{192} + \frac{5 x}{96} - \frac{1}{64} & - \frac{7 x^{2}}{192} + \frac{5 x}{96} - \frac{1}{64} & \frac{x \left(7 x - 3\right)}{36} & \frac{x^{2}}{64} - \frac{x}{32} + \frac{1}{64} & \frac{x^{2}}{64} - \frac{x}{32} + \frac{1}{64} & \frac{x \left(1 - x\right)}{12}
\end{pmatrix}.
$}
\end{equation}
The matrix $M_5$ is defined as the principal submatrix of $M$ obtained by taking the rows and columns whose indices lie in $T_5$.
\begin{equation}
\resizebox{\textwidth}{!}{$
\begin{pmatrix}
- \frac{x^{2}}{192} - \frac{x}{96} + \frac{1}{64} & \frac{x^{2}}{64} - \frac{x}{32} + \frac{1}{64} & \frac{x^{2}}{64} - \frac{x}{32} + \frac{1}{64} & \frac{x \left(- x - 3\right)}{72} & \frac{x \left(x - 1\right)}{24} & \frac{x \left(x - 1\right)}{24}\\
\frac{x^{2}}{64} - \frac{x}{32} + \frac{1}{64} & - \frac{x^{2}}{192} - \frac{x}{96} + \frac{1}{64} & \frac{x^{2}}{64} - \frac{x}{32} + \frac{1}{64} & \frac{x \left(x - 1\right)}{24} & \frac{x \left(- x - 3\right)}{72} & \frac{x \left(x - 1\right)}{24}\\
\frac{x^{2}}{64} - \frac{x}{32} + \frac{1}{64} & \frac{x^{2}}{64} - \frac{x}{32} + \frac{1}{64} & \frac{x \left(1 - x\right)}{12} & \frac{x \left(x - 1\right)}{24} & \frac{x \left(x - 1\right)}{24} & - \frac{2 x^{2}}{9}\\
\frac{x \left(- x - 3\right)}{72} & \frac{x \left(x - 1\right)}{24} & \frac{x \left(x - 1\right)}{24} & - \frac{x^{2}}{192} - \frac{x}{96} + \frac{1}{64} & \frac{x^{2}}{64} - \frac{x}{32} + \frac{1}{64} & \frac{x^{2}}{64} - \frac{x}{32} + \frac{1}{64}\\
\frac{x \left(x - 1\right)}{24} & \frac{x \left(- x - 3\right)}{72} & \frac{x \left(x - 1\right)}{24} & \frac{x^{2}}{64} - \frac{x}{32} + \frac{1}{64} & - \frac{x^{2}}{192} - \frac{x}{96} + \frac{1}{64} & \frac{x^{2}}{64} - \frac{x}{32} + \frac{1}{64}\\
\frac{x \left(x - 1\right)}{24} & \frac{x \left(x - 1\right)}{24} & - \frac{2 x^{2}}{9} & \frac{x^{2}}{64} - \frac{x}{32} + \frac{1}{64} & \frac{x^{2}}{64} - \frac{x}{32} + \frac{1}{64} & \frac{x \left(1 - x\right)}{12}
\end{pmatrix}.
$}
\end{equation}
The matrix $M_6$ is defined as the principal submatrix of $M$ obtained by taking the rows and columns whose indices lie in $T_6$.
\begin{equation}
\resizebox{\textwidth}{!}{$
\begin{pmatrix}
- \frac{x^{2}}{192} - \frac{x}{96} + \frac{1}{64} & \frac{x \left(- x - 3\right)}{72} & \frac{x^{2}}{64} - \frac{x}{32} + \frac{1}{64} & \frac{x \left(x - 1\right)}{24} & \frac{x^{2}}{64} - \frac{x}{32} + \frac{1}{64} & \frac{x \left(x - 1\right)}{24}\\
\frac{x \left(- x - 3\right)}{72} & - \frac{x^{2}}{192} - \frac{x}{96} + \frac{1}{64} & \frac{x \left(x - 1\right)}{24} & \frac{x^{2}}{64} - \frac{x}{32} + \frac{1}{64} & \frac{x \left(x - 1\right)}{24} & \frac{x^{2}}{64} - \frac{x}{32} + \frac{1}{64}\\
\frac{x^{2}}{64} - \frac{x}{32} + \frac{1}{64} & \frac{x \left(x - 1\right)}{24} & - \frac{x^{2}}{192} - \frac{x}{96} + \frac{1}{64} & \frac{x \left(- x - 3\right)}{72} & \frac{x^{2}}{64} - \frac{x}{32} + \frac{1}{64} & \frac{x \left(x - 1\right)}{24}\\
\frac{x \left(x - 1\right)}{24} & \frac{x^{2}}{64} - \frac{x}{32} + \frac{1}{64} & \frac{x \left(- x - 3\right)}{72} & - \frac{x^{2}}{192} - \frac{x}{96} + \frac{1}{64} & \frac{x \left(x - 1\right)}{24} & \frac{x^{2}}{64} - \frac{x}{32} + \frac{1}{64}\\
\frac{x^{2}}{64} - \frac{x}{32} + \frac{1}{64} & \frac{x \left(x - 1\right)}{24} & \frac{x^{2}}{64} - \frac{x}{32} + \frac{1}{64} & \frac{x \left(x - 1\right)}{24} & \frac{x \left(1 - x\right)}{12} & - \frac{2 x^{2}}{9}\\
\frac{x \left(x - 1\right)}{24} & \frac{x^{2}}{64} - \frac{x}{32} + \frac{1}{64} & \frac{x \left(x - 1\right)}{24} & \frac{x^{2}}{64} - \frac{x}{32} + \frac{1}{64} & - \frac{2 x^{2}}{9} & \frac{x \left(1 - x\right)}{12}
\end{pmatrix}.
$}
\end{equation}
The matrix $M_7$ is defined as the principal submatrix of $M$ obtained by taking the rows and columns whose indices lie in $T_7$.
\begin{equation}
\resizebox{\textwidth}{!}{$
\begin{pmatrix}
- \frac{x^{2}}{192} - \frac{x}{96} + \frac{1}{64} & \frac{x^{2}}{64} - \frac{x}{32} + \frac{1}{64} & \frac{x^{2}}{64} - \frac{x}{32} + \frac{1}{64} & \frac{x \left(- x - 3\right)}{72} & \frac{x \left(x - 1\right)}{24} & \frac{x \left(x - 1\right)}{24}\\
\frac{x^{2}}{64} - \frac{x}{32} + \frac{1}{64} & - \frac{x^{2}}{192} - \frac{x}{96} + \frac{1}{64} & \frac{x^{2}}{64} - \frac{x}{32} + \frac{1}{64} & \frac{x \left(x - 1\right)}{24} & \frac{x \left(- x - 3\right)}{72} & \frac{x \left(x - 1\right)}{24}\\
\frac{x^{2}}{64} - \frac{x}{32} + \frac{1}{64} & \frac{x^{2}}{64} - \frac{x}{32} + \frac{1}{64} & \frac{x \left(1 - x\right)}{12} & \frac{x \left(x - 1\right)}{24} & \frac{x \left(x - 1\right)}{24} & - \frac{2 x^{2}}{9}\\
\frac{x \left(- x - 3\right)}{72} & \frac{x \left(x - 1\right)}{24} & \frac{x \left(x - 1\right)}{24} & - \frac{x^{2}}{192} - \frac{x}{96} + \frac{1}{64} & \frac{x^{2}}{64} - \frac{x}{32} + \frac{1}{64} & \frac{x^{2}}{64} - \frac{x}{32} + \frac{1}{64}\\
\frac{x \left(x - 1\right)}{24} & \frac{x \left(- x - 3\right)}{72} & \frac{x \left(x - 1\right)}{24} & \frac{x^{2}}{64} - \frac{x}{32} + \frac{1}{64} & - \frac{x^{2}}{192} - \frac{x}{96} + \frac{1}{64} & \frac{x^{2}}{64} - \frac{x}{32} + \frac{1}{64}\\
\frac{x \left(x - 1\right)}{24} & \frac{x \left(x - 1\right)}{24} & - \frac{2 x^{2}}{9} & \frac{x^{2}}{64} - \frac{x}{32} + \frac{1}{64} & \frac{x^{2}}{64} - \frac{x}{32} + \frac{1}{64} & \frac{x \left(1 - x\right)}{12}
\end{pmatrix}.
$}
\end{equation}
The matrix $M_8$ is defined as the principal submatrix of $M$ obtained by taking the rows and columns whose indices lie in $T_8$.
\begin{equation}
\resizebox{0.85\textwidth}{!}{$
\begin{pmatrix}
\frac{x^{2}}{64} - \frac{x}{32} + \frac{1}{64} & - \frac{7 x^{2}}{192} + \frac{5 x}{96} - \frac{1}{64} & - \frac{7 x^{2}}{192} + \frac{5 x}{96} - \frac{1}{64} & \frac{49 x^{2}}{576} - \frac{7 x}{96} + \frac{1}{64}\\
- \frac{7 x^{2}}{192} + \frac{5 x}{96} - \frac{1}{64} & \frac{x^{2}}{64} - \frac{x}{32} + \frac{1}{64} & \frac{49 x^{2}}{576} - \frac{7 x}{96} + \frac{1}{64} & - \frac{7 x^{2}}{192} + \frac{5 x}{96} - \frac{1}{64}\\
- \frac{7 x^{2}}{192} + \frac{5 x}{96} - \frac{1}{64} & \frac{49 x^{2}}{576} - \frac{7 x}{96} + \frac{1}{64} & \frac{x^{2}}{64} - \frac{x}{32} + \frac{1}{64} & - \frac{7 x^{2}}{192} + \frac{5 x}{96} - \frac{1}{64}\\
\frac{49 x^{2}}{576} - \frac{7 x}{96} + \frac{1}{64} & - \frac{7 x^{2}}{192} + \frac{5 x}{96} - \frac{1}{64} & - \frac{7 x^{2}}{192} + \frac{5 x}{96} - \frac{1}{64} & \frac{x^{2}}{64} - \frac{x}{32} + \frac{1}{64}
\end{pmatrix}.
$}
\end{equation}
The matrix $M_9$ is defined as the principal submatrix of $M$ obtained by taking the rows and columns whose indices lie in $T_9$.
\begin{equation}
\resizebox{0.85\textwidth}{!}{$
\begin{pmatrix}
\frac{x^{2}}{64} - \frac{x}{32} + \frac{1}{64} & \frac{x \left(x - 1\right)}{24} & - \frac{7 x^{2}}{192} + \frac{5 x}{96} - \frac{1}{64} & \frac{x \left(3 - 7 x\right)}{72}\\
\frac{x \left(x - 1\right)}{24} & \frac{x^{2}}{64} - \frac{x}{32} + \frac{1}{64} & \frac{x \left(3 - 7 x\right)}{72} & - \frac{7 x^{2}}{192} + \frac{5 x}{96} - \frac{1}{64}\\
- \frac{7 x^{2}}{192} + \frac{5 x}{96} - \frac{1}{64} & \frac{x \left(3 - 7 x\right)}{72} & \frac{x^{2}}{64} - \frac{x}{32} + \frac{1}{64} & \frac{x \left(x - 1\right)}{24}\\
\frac{x \left(3 - 7 x\right)}{72} & - \frac{7 x^{2}}{192} + \frac{5 x}{96} - \frac{1}{64} & \frac{x \left(x - 1\right)}{24} & \frac{x^{2}}{64} - \frac{x}{32} + \frac{1}{64}
\end{pmatrix}.
$}
\end{equation}
The matrix $M_{10}$ is defined as the principal submatrix of $M$ obtained by taking the rows and columns whose indices lie in $T_{10}$.
\begin{equation}
\resizebox{0.85\textwidth}{!}{$
\begin{pmatrix}
\frac{x^{2}}{64} - \frac{x}{32} + \frac{1}{64} & - \frac{7 x^{2}}{192} + \frac{5 x}{96} - \frac{1}{64} & \frac{x \left(x - 1\right)}{24} & \frac{x \left(3 - 7 x\right)}{72}\\
- \frac{7 x^{2}}{192} + \frac{5 x}{96} - \frac{1}{64} & \frac{x^{2}}{64} - \frac{x}{32} + \frac{1}{64} & \frac{x \left(3 - 7 x\right)}{72} & \frac{x \left(x - 1\right)}{24}\\
\frac{x \left(x - 1\right)}{24} & \frac{x \left(3 - 7 x\right)}{72} & \frac{x^{2}}{64} - \frac{x}{32} + \frac{1}{64} & - \frac{7 x^{2}}{192} + \frac{5 x}{96} - \frac{1}{64}\\
\frac{x \left(3 - 7 x\right)}{72} & \frac{x \left(x - 1\right)}{24} & - \frac{7 x^{2}}{192} + \frac{5 x}{96} - \frac{1}{64} & \frac{x^{2}}{64} - \frac{x}{32} + \frac{1}{64}
\end{pmatrix}.
$}
\end{equation}

The matrix $M_{11}$ is defined as the principal submatrix of $M$ obtained by taking the rows and columns whose indices lie in $T_{11}$.
\begin{equation}
\resizebox{0.85\textwidth}{!}{$
\begin{pmatrix}
\frac{x^{2}}{64} - \frac{x}{32} + \frac{1}{64} & \frac{x \left(x - 1\right)}{24} & \frac{x \left(x - 1\right)}{24} & \frac{x^{2}}{9}\\
\frac{x \left(x - 1\right)}{24} & \frac{x^{2}}{64} - \frac{x}{32} + \frac{1}{64} & \frac{x^{2}}{9} & \frac{x \left(x - 1\right)}{24}\\
\frac{x \left(x - 1\right)}{24} & \frac{x^{2}}{9} & \frac{x^{2}}{64} - \frac{x}{32} + \frac{1}{64} & \frac{x \left(x - 1\right)}{24}\\
\frac{x^{2}}{9} & \frac{x \left(x - 1\right)}{24} & \frac{x \left(x - 1\right)}{24} & \frac{x^{2}}{64} - \frac{x}{32} + \frac{1}{64}
\end{pmatrix}.
$}
\end{equation}

The matrix $M_{12}$ is defined as the principal submatrix of $M$ obtained by taking the rows and columns whose indices lie in $T_{12}$.
\begin{equation}
\resizebox{0.85\textwidth}{!}{$
\begin{pmatrix}
\frac{x^{2}}{64} - \frac{x}{32} + \frac{1}{64} & \frac{x \left(x - 1\right)}{24} & - \frac{7 x^{2}}{192} + \frac{5 x}{96} - \frac{1}{64} & \frac{x \left(3 - 7 x\right)}{72}\\
\frac{x \left(x - 1\right)}{24} & \frac{x^{2}}{64} - \frac{x}{32} + \frac{1}{64} & \frac{x \left(3 - 7 x\right)}{72} & - \frac{7 x^{2}}{192} + \frac{5 x}{96} - \frac{1}{64}\\
- \frac{7 x^{2}}{192} + \frac{5 x}{96} - \frac{1}{64} & \frac{x \left(3 - 7 x\right)}{72} & \frac{x^{2}}{64} - \frac{x}{32} + \frac{1}{64} & \frac{x \left(x - 1\right)}{24}\\
\frac{x \left(3 - 7 x\right)}{72} & - \frac{7 x^{2}}{192} + \frac{5 x}{96} - \frac{1}{64} & \frac{x \left(x - 1\right)}{24} & \frac{x^{2}}{64} - \frac{x}{32} + \frac{1}{64}
\end{pmatrix}.
$}
\end{equation}

The matrix $M_{13}$ is defined as the principal submatrix of $M$ obtained by taking the rows and columns whose indices lie in $T_{13}$.
\begin{equation}
\resizebox{0.85\textwidth}{!}{$
\begin{pmatrix}
\frac{x^{2}}{64} - \frac{x}{32} + \frac{1}{64} & \frac{x \left(x - 1\right)}{24} & \frac{x \left(x - 1\right)}{24} & \frac{x^{2}}{9}\\
\frac{x \left(x - 1\right)}{24} & \frac{x^{2}}{64} - \frac{x}{32} + \frac{1}{64} & \frac{x^{2}}{9} & \frac{x \left(x - 1\right)}{24}\\
\frac{x \left(x - 1\right)}{24} & \frac{x^{2}}{9} & \frac{x^{2}}{64} - \frac{x}{32} + \frac{1}{64} & \frac{x \left(x - 1\right)}{24}\\
\frac{x^{2}}{9} & \frac{x \left(x - 1\right)}{24} & \frac{x \left(x - 1\right)}{24} & \frac{x^{2}}{64} - \frac{x}{32} + \frac{1}{64}
\end{pmatrix}.
$}
\end{equation}

The matrix $M_{14}$ is defined as the principal submatrix of $M$ obtained by taking the rows and columns whose indices lie in $T_{14}$.
\begin{equation}
\resizebox{0.85\textwidth}{!}{$
\begin{pmatrix}
\frac{x^{2}}{64} - \frac{x}{32} + \frac{1}{64} & - \frac{7 x^{2}}{192} + \frac{5 x}{96} - \frac{1}{64} & \frac{x \left(x - 1\right)}{24} & \frac{x \left(3 - 7 x\right)}{72}\\
- \frac{7 x^{2}}{192} + \frac{5 x}{96} - \frac{1}{64} & \frac{x^{2}}{64} - \frac{x}{32} + \frac{1}{64} & \frac{x \left(3 - 7 x\right)}{72} & \frac{x \left(x - 1\right)}{24}\\
\frac{x \left(x - 1\right)}{24} & \frac{x \left(3 - 7 x\right)}{72} & \frac{x^{2}}{64} - \frac{x}{32} + \frac{1}{64} & - \frac{7 x^{2}}{192} + \frac{5 x}{96} - \frac{1}{64}\\
\frac{x \left(3 - 7 x\right)}{72} & \frac{x \left(x - 1\right)}{24} & - \frac{7 x^{2}}{192} + \frac{5 x}{96} - \frac{1}{64} & \frac{x^{2}}{64} - \frac{x}{32} + \frac{1}{64}
\end{pmatrix}.
$}
\end{equation}

The matrix $M_{15}$ is defined as the principal submatrix of $M$ obtained by taking the rows and columns whose indices lie in $T_{15}$.
\begin{equation}
\resizebox{0.85\textwidth}{!}{$
\begin{pmatrix}
\frac{x^{2}}{64} - \frac{x}{32} + \frac{1}{64} & \frac{x \left(x - 1\right)}{24} & \frac{x \left(x - 1\right)}{24} & \frac{x^{2}}{9}\\
\frac{x \left(x - 1\right)}{24} & \frac{x^{2}}{64} - \frac{x}{32} + \frac{1}{64} & \frac{x^{2}}{9} & \frac{x \left(x - 1\right)}{24}\\
\frac{x \left(x - 1\right)}{24} & \frac{x^{2}}{9} & \frac{x^{2}}{64} - \frac{x}{32} + \frac{1}{64} & \frac{x \left(x - 1\right)}{24}\\
\frac{x^{2}}{9} & \frac{x \left(x - 1\right)}{24} & \frac{x \left(x - 1\right)}{24} & \frac{x^{2}}{64} - \frac{x}{32} + \frac{1}{64}
\end{pmatrix}.
$}
\end{equation}

The matrix $M_{16}$ is defined as the principal submatrix of $M$ obtained by taking the rows and columns whose indices lie in $T_{16}$.
\begin{equation}
\resizebox{0.85\textwidth}{!}{$
\begin{pmatrix}
\frac{x^{2}}{64} - \frac{x}{32} + \frac{1}{64} & \frac{x \left(x - 1\right)}{24} & \frac{x \left(x - 1\right)}{24} & \frac{x^{2}}{9}\\
\frac{x \left(x - 1\right)}{24} & \frac{x^{2}}{64} - \frac{x}{32} + \frac{1}{64} & \frac{x^{2}}{9} & \frac{x \left(x - 1\right)}{24}\\
\frac{x \left(x - 1\right)}{24} & \frac{x^{2}}{9} & \frac{x^{2}}{64} - \frac{x}{32} + \frac{1}{64} & \frac{x \left(x - 1\right)}{24}\\
\frac{x^{2}}{9} & \frac{x \left(x - 1\right)}{24} & \frac{x \left(x - 1\right)}{24} & \frac{x^{2}}{64} - \frac{x}{32} + \frac{1}{64}
\end{pmatrix}.
$}
\end{equation}

\end{document}